\documentclass{LMCS}

\def\doi{8(3:14)2012}
\lmcsheading%
{\doi}
{1--32}
{}
{}
{Jul.~19, 2006}
{Sep.~12, 2012}
{}
 
\usepackage{enumerate,hyperref}
\usepackage[all]{xy} 
\usepackage{bbm} %Blackbord 
\usepackage{bm} %bold math
\usepackage{latexsym} 
\usepackage{amsmath} 
\usepackage{amssymb} 
\usepackage{amsthm} 
\usepackage{url} 
\usepackage{ifthen} 
 
\newcommand{\twobb}{\mathbbm{2}}   
\newcommand{\Nbb}{\mathbbm{N}}

\newcommand{\npskip}{\vspace{-\parskip}}

\newcommand{\pskip}{\par\vspace{2ex plus 0.5ex minus 0.2ex}\vspace{-\parskip}\noindent}

\renewcommand{\phi}{\varphi}

\newcommand{\biimp}{\ensuremath{\leftrightarrow}} 
\newcommand{\et}{\ensuremath{\wedge}} %\and is already a latex command 
\newcommand{\too}{\longrightarrow}

\newcommand{\eps}{\varepsilon}

\newcommand{\acal}{\mathcal{ A}}  
\newcommand{\bcal}{\mathcal{ B}}  
\newcommand{\ccal}{\mathcal{ C}}  
\newcommand{\dcal}{\mathcal{ D}}

\newcommand{\hcal}{\mathcal{ H}}  
\newcommand{\ical}{\mathcal{ I}}

\newcommand{\lcal}{\mathcal{ L}}

\newcommand{\pcal}{\mathcal{ P}}

\newcommand{\scal}{\mathcal{ S}}  
\newcommand{\tcal}{\mathcal{ T}}

\newcommand{\xcal}{\mathcal{ X}}

\newcommand{\Poset}{\mathsf{Poset}} 
 
\newtheorem{theorem}{Theorem}[section] 
\newtheorem{theorem-conj}{Theorem (Conjecture)}[section] 
 
\newtheorem{lemma-conj}[theorem]{Lemma (Conjecture)}

\theoremstyle{definition} 
\newtheorem{definition}[theorem]{Definition} 
\newtheorem{example}[theorem]{Example}

\theoremstyle{remark}

\newcommand{\id}{\mathrm{id}} 
\newcommand{\iso}{\cong} 
\newcommand{\op}{^{\mathrm{op}}} 
 
\newcommand{\Set}{\mathsf{Set}} 
\newcommand{\Fin}{\mathsf{Fin}} 
\newcommand{\Ab}{\mathsf{Ab}} 
\newcommand{\Vect}{\mathsf{Vect}}

\newcommand{\BA}{\mathsf{BA}} 
\newcommand{\DL}{\mathsf{DL}} 
 
\newcommand{\CABA}{\mathsf{CABA}} 
\newcommand{\Stone}{\mathsf{Stone}}

\newcommand{\Ind}{\mathsf{Ind}} 
\newcommand{\Pro}{\mathsf{Pro}} 
 
\newcommand{\Alg}{\mathsf{Alg}}

\newcommand{\Coalg}{\mathsf{Coalg}} % usage: \Coalg(T) 
 
\newcommand{\Id}{\mathit{Id}} 
\newcommand{\Pow}{\mathit{\pcal}}

\theoremstyle{remark}

\theoremstyle{definition}

\newcommand{\adj}{\mathrel{\dashv}}

\renewcommand{\rho}{\varrho}

\newcommand{\onebb}{\mathbbm{1}}

\newcommand{\fp}{{{\mathit{fp}}}} 
\newcommand{\sfp}{{{\mathit{sfp}}}}

\newcommand{\lsem}{\mathopen{[\![}}    
\newcommand{\rsem}{\mathclose{]\!]}}    
\newcommand{\sem}[1]{\lsem #1 \rsem}

\begin{document} 
 
\title[Strongly Complete Logics for Coalgebras]{Strongly Complete Logics for Coalgebras} 
 
\author[Kurz]{Alexander Kurz\rsuper a}
\address{{\lsuper a}University of Leicester, UK}
\email{kurz@mcs.le.ac.uk}
\thanks{{\lsuper a}Supported by Nuffield Foundation Grant NUF-NAL04.}

\author[Rosick{\'y}]{Ji\v{r}\'\i{} Rosick{\'y}\rsuper b}
\address{{\lsuper b}Masaryk University, Brno, Czech Republic}
\email{rosicky@math.muni.cz} 
\thanks{{\lsuper b}Supported by the Ministry
  of Education of the Czech Republic under the project 1M0545.} 
  
\begin{abstract}  
  \noindent Coalgebras for a functor model different types of
  transition systems in a uniform way. This paper focuses on a uniform
  account of finitary logics for set-based coalgebras. In particular,
  a general construction of a logic from an arbitrary set-functor is
  given and proven to be strongly complete under additional
  assumptions. We proceed in three parts.

  Part I argues that sifted colimit preserving functors are those
  functors that preserve universal algebraic structure. Our main
  theorem here states that a functor preserves sifted colimits if and
  only if it has a finitary presentation by operations and equations.
  Moreover, the presentation of the category of algebras for the
  functor is obtained compositionally from the presentations of the
  underlying category and of the functor.
  
  Part II investigates algebras for a functor over ind-completions
  and extends the theorem of J{\'o}nsson and Tarski on canonical
  extensions of Boolean algebras with operators to this setting.
  
  Part III shows, based on Part I, how to associate a finitary
  logic to any finite-sets preserving functor $T$. Based on Part II we
  prove the logic to be strongly complete under a reasonable condition
  on $T$.
\end{abstract}  
 
\subjclass{F.3.2; F.4.1}
\keywords{coalgebras, modal logic, Stone duality, algebraic theories, sifted
colimits, presentation of functors}

\maketitle \newpage
  
\tableofcontents

\section{Introduction} 

\noindent This paper can be read as consisting of three independent
parts or it can be read with a unifying story in mind. Since the three
parts may be of interest to different readers, and require somewhat
different prerequisites, we keep them separated and try to make them
reasonably self-contained. On the other hand, the story will be of
interest to some readers and so we sketch it in this introduction.  We
begin with a brief overview of the three parts.

\medskip\noindent Part I presents an investigation in categorical
universal algebra. In universal algebra, a variety is a category that
has a finitary presentation by operations and equations.  We
investigate functors on varieties that have finitary presentations by
operations and equations. We will show that an endofunctor on a
variety has such a presentation if and only if the functor preserves
sifted colimits.
 
\medskip\noindent Part II studies a topic in Stone duality. Given a
small finitely complete and cocomplete category $\ccal$, one obtains a
dual adjunction between the ind-completions $\Ind\ccal$ and
$\Ind(\ccal\op)$. We prove a J\'onsson-Tarski style representation
theorem showing how to represent algebras over $\Ind\ccal$ as duals of
coalgebras over $\Ind(\ccal\op)$.

\medskip\noindent Part III investigates how to associate in a uniform
way to a set-functor $T$ a logic for $T$-coalgebras. Using the results
from Part I and Part II, we show that, under a mild additional
condition, any finite set-preserving functor $T$ has a strongly
complete finitary modal logic.

\bigskip Our story starts with the idea of universal coalgebra as a
general theory of systems, due to Rutten \cite{rutten:uc-j}, which
allows us to deal with issues such as bisimilarity, coinduction, etc
in a uniform way. A natural question then is whether something similar
can be done for logics of coalgebras. The first answer to this was
Moss's seminal Coalgebraic Logic~\cite{moss:cl}. Moss's logic is
parametric in $T$, the basic idea being to take $T$ itself as an
operation to construct formulas: if $\Phi$ is a set of formulas, then
$T(\Phi)$ is a set of formulas.

\medskip\noindent Following on from Moss~\cite{moss:cl}, attention
turned to the question of how to set up logics for coalgebras in which
formulas are built according to a more conventional scheme: If $\Box$
is a unary operation symbol, or `modal operator', and $\phi$ is a formula,
then $\Box\phi$ is a formula. After some work in this direction,
see eg \cite{kurz:cmcs98-j,roessiger:ml98-j,jacobs:many-sorted},
Pattinson~\cite{pattinson:cml-j} showed that such languages arise from
modal operators given by natural transformations $2^X\to 2^{TX}$,
which are called predicate liftings of
$T$. Schr\"oder~\cite{schroeder:fossacs05} investigated the logics
given by all predicate liftings of finite arity and showed that these
logics are expressive for finitary functors $T$.

\medskip\noindent Another approach is based on Stone duality. In the
context of coalgebraic logic, it was first advocated by
\cite{bons-kurz:fossacs05}, but is based on the ideas of domain theory
in logical form \cite{abramsky:dtlf,abra-jung:dt}. We will explain it
briefly here.

\bigskip We think of Stone duality
\cite{johnstone:stone-spaces,abra-jung:dt} as relating a category of
algebras $\acal$, representing a propositional logic, to a category of
topological spaces $\xcal$, representing the state-based models of the
logic. The duality is provided by two contravariant functors $P$ and
$S$,
\begin{equation} 
\label{eins} 
\xymatrix{ 
  {\ \xcal} \ar@/^/[r]^{P} & {\acal \ } \ar@/^/[l]^{S} . 
   } 
\end{equation} 
$P$ maps a space $X$ to a propositional theory and $S$ maps a
propositional theory to its `canonical model'. 
For the moment let us assume that $\acal$ and $\xcal$ are dually
equivalent, as it is the case when $\xcal$ is the category $\Stone$ of
Stone spaces and $\acal$ is the category $\BA$ of Boolean
algebras. This means that, from an abstract categorical point of view,
the two categories are the same, up to reversal of arrows. But this
ignores the extra structure which consists of both categories
having a forgetful functor to $\Set$, with $\xcal\to\Set$ not being dual
to $\acal\to\Set$. An object of $\xcal$ specifies a \emph{set} of
states that serves as our semantic domain; an object of $\acal$
specifies a \emph{set} of propositions, which we use to specify
properties of spaces. More specifically, we will assume that $\acal$ is a
variety, that is, $\acal$ is isomorphic to a category of algebras given by
operations and equations, the operations being our logical connectives
and the equations the logical axioms. Or, equivalently, we can say
that $\acal\to\Set$ is finitary and monadic.

\medskip\noindent To extend a basic Stone duality as above to
coalgebras over $\xcal$, we consider, as Abramsky did in his Domain
Theory in Logical Form \cite{abramsky:dtlf}, the dual $L$ of $T$:
\begin{equation}\label{zwei} 
  \xymatrix{ 
    {\ \xcal} \ar@(dl,ul)[]^{T} \ar@/^/[r]^{P} & {\acal \ } \ar@/^/[l]^{S} 
    \ar@(dr,ur)[]_{L} }  \quad\quad\quad\quad LP\iso PT
\end{equation} 
Then the category of $L$-algebras is dual to the category of
$T$-coalgebras and the initial $L$-algebra provides a propositional
theory characterising $T$-bisimilarity. Moreover, if $L$ can be
presented by generators and relations, or rather operations and
equations, one inherits a proof system from equational logic
which is sound and strongly complete. Thus, \emph{logics for
  $T$-coalgebras arise from presentations of the dual of $T$ by
  operations and equations}. Part I characterises those
functors $L$ on varieties $\acal$ that have a finitary presentation by operations and equations.

\medskip The approach indicated in Diagram~\ref{zwei} can be applied
to set-coalgebras, but as the dual of $\Set$ is the category $\CABA$
of \emph{complete} atomic Boolean algebras, the corresponding logics
would become infinitary.  Hence, being interested in finitary logics,
we are led to consider two Stone dualities
\begin{equation}\label{drei} 
  \xymatrix{ 
    {\Stone }  \ar@/^/[r]^{} %\ar@/^/[d]^{} 
    &  
      {\BA} \ar@/^/[l]^{}\ar@(dr,ur)[]_{L} %\ar@/^/[d]^{} 
    \\ 
    {\Set} \ar@(dl,ul)[]^{T} \ar@/^/[r]^{} %\ar@/^/[u]^{} 
    &  
      {\CABA} \ar@/^/[l]^{}%\ar@/^/[u]^{} 
}  
\end{equation} 
The upper row is the duality between Stone spaces and Boolean
algebras, accounting for (classical finitary) propositional logic. 
The lower row is the duality where our set-based $T$-coalgebras
live. How can the two be related?

\medskip\noindent The crucial observation is the following. $\BA$ is
the ind-completion of the category $\BA_\omega$ of finite Boolean
algebras, that is, the completion of finite Boolean algebras under
filtered colimits; $\Set$ is the ind-completion of the category
$\Set_\omega$ of finite sets; and finite sets are dually equivalent to
finite Boolean algebras. In other words, $\Set\op$ is the
pro-completion of finite Boolean algebras, that is, the completion of
finite Boolean algebras under cofiltered limits.
\begin{equation}
  \xymatrix@C=15pt{ 
    {\BA} \ar@(dl,ul)[]^{L \ } 
    \ar@/_/[rr]_{S}  & &  
    {\Set\op}\ar@/_/[ll]_{P} \ar@(dr,ur)[]_{\ T\op}  
    \\ 
    & \ar[ul]^{\hat{(-)}} \BA_\omega\simeq\Set_\omega\op  \ar[ur]_{\bar{(-)}} & 
  } 
\end{equation} 
This gives us a systematic link between Boolean propositional logic
and its set-theoretic semantic. It is extended to modal logics and
their coalgebraic semantics by a natural transformation
\[\delta:LP \to PT \] lifting $P$ to a functor $\tilde
P:\Coalg(T)\to\Alg(L)$. Not in general, but in a large number of
important examples, we also obtain a (not necessarily natural)
transformation
\[h:SL\to TS\] giving rise to a map on objects $\tilde
S:\Alg(L)\to\Coalg(T)$.  Part II shows that then every $L$-algebra $A$
can be embedded into the $P$-image of a $T$-coalgebra, namely, into
$\tilde P\tilde S A$.  This extends the J\'onsson-Tarski theorem from
Kripke frames \cite[Thm 5.43]{brv:ml} to coalgebras and shows that the
logics $L$ are canonical in the sense that all formulas hold in the
underlying frame of the canonical model $\tilde S I$, where $I$ is the
Lindenbaum-algebra of $L$.

\medskip Starting from an arbitrary functor $T:\Set\to\Set$, Part III shows first how to define a suitable  $L$.  From Part I, we know that $L$ has a
presentation by operations and equations and therefore corresponds to
a `concrete' modal logic in the traditional sense. This is detailed in
Section~\ref{sec:ml-abs-concr}. Section~\ref{sec:ml-set} then applies
the J\'onsson-Tarski theorem from Part II to obtain strong
completeness results for modal logics for coalgebras.

\bigskip\noindent \textbf{Further related work } Unary predicate
liftings and a criterion for weak completeness go back to Pattinson
\cite{pattinson:cml-j}. The observation that all logics given by
predicate liftings correspond to a functor $L$ on $\BA$ was made in
\cite{kkp:cmcs04}, with the approach of functorial modal logic going
back to \cite{bons-kurz:fossacs05}. The logic of all predicate
liftings of finite arity was introduced in
Schr\"oder~\cite{schroeder:fossacs05}. Our notion of a presentation of
a functor and the fact that such presentations give rise to modal
logics is from \cite{bons-kurz:fossacs06}. The J\'onsson-Tarski
theorem of this paper generalises the corresponding theorem in
\cite{kkp:calco05}. The process of taking a finite set preserving
functor and extending it to $\BA$, and hence to $\Stone$, is related
to a construction in Worrell~\cite{worrell:final-sequence} where a
set-functor is lifted to complete ultrametric
spaces. Klin~\cite{klin:mfps07} generalises the expressivity result of
\cite{schroeder:fossacs05} working essentially with the same
adjunction as in Diagram~\ref{zwei}.

Since an earlier version of the paper was made available in June 2006, the field
continued to develop quickly. For example,
Schr\"oder~\cite{schroeder:finite-model-construction} improves on
\cite{pattinson:cml-j} by showing that complete axiomatisations always
exist and, moreover, that completeness holds wrt finite models.
\cite{patt-schr:algebraic-semantics} uses Stone duality and algebraic
techniques to derive conditions for the finite model property of
logics with additional non-rank-1 axioms. Schr\"oder and
Pattinson~\cite{schr-patt:strong-completeness} push the strong
completeness result of this paper further and add several important
examples.

\pskip\textbf{Acknowledgements } The first author is grateful to
Adriana Balan, Nick Bezhanishvili, Marcello Bonsangue, Corina C\^\i
rstea, Bartek Klin, Clemens Kupke, Tadeusz Litak, Rob Myers, Dirk
Pattinson, Daniela Petri{\c s}an, Katsuhiko Sano, Vincent Schmitt,
Lutz Schr\"oder, Ji\v{r}\'\i Velebil, Yde Venema, and James Worrell who all
contributed to aspects of this work.

\section{Introduction to Part I: Algebras and
  Varieties}\label{sec:prel}
\noindent
There is a general agreement that algebras over a category $\acal$ are
described by means of a monad $M:\acal\to\acal$ (see \cite{ML}). In
the case when $\acal$ is the category $\Set$ of sets, the category
$\Set_M$ of algebras over a monad $M:\Set\to\Set$ always has a
\emph{presentation} in the sense that there exists a signature
$\Sigma$ (allowing infinite arities) and equations $E$ such that $\Set_M$ is concretely isomorphic
to the category $\Alg(\Sigma,E)$ of $(\Sigma,E)$-algebras. Both
$\Sigma$ and $E$ can be proper classes but the characteristic property
is that free $(\Sigma,E)$-algebras always exist. This allows compact
Hausdorff spaces and complete atomic Boolean algebras but eliminates
complete Boolean algebras.  The important special case is when a monad
$M$ has a rank which means that it preserves $\lambda$-filtered
colimits for some regular cardinal $\lambda$. It corresponds to the
case when $\Sigma$ is a set (and $\lambda$ is greater than arities of
all $\Sigma$-operations).  In particular, algebras over a monad
preserving filtered colimits (such functors are called
\emph{finitary}) then correspond to classical universal algebras.

Alternatively and equivalently, classical universal algebras can be
described by algebraic theories (see Lawvere \cite{lawvere:thesis}).
Recall that an \emph{algebraic theory} is a category $\tcal$ whose
objects are integers $0,1,2,\dots$ and such that $n$ is a product of
$n$ copies of $1$ for each $n=0,1,\dots$. It means that $\tcal$ has
finite products and, in particular, a terminal object $0$.
Intuitively, arrows $n\to 1$ represent $n$-ary terms and commuting
diagrams represent equations. An algebraic theory $\tcal$ determines
the category $\Alg(\tcal)$ which is the full subcategory of
$\Set^\tcal$ consisting of all functors $\tcal\to\Set$ preserving
finite products. The underlying functor $\Alg(\tcal)\to\Set$ is given
by the evaluation at $1$. A category is called a \emph{variety} if it
is equivalent to $\Alg(\tcal)$ for some algebraic theory $\tcal$.

Given a category $\acal$ with finite products and an algebraic theory
$\tcal$, one can still define $\tcal$-algebras in $\acal$ as finite
product preserving functors $\tcal\to\acal$ and consider the category
$\Alg_\acal(\tcal)$ of these algebras. In general, there is no
guarantee that this category is monadic over $\acal$ because free
algebras do not need to exist. But, whenever $\acal$ is locally
presentable (see \cite{ar}), $\Alg_\acal(\tcal)$ is always monadic and
the corresponding monad $M:\acal\to\acal$ has a rank. In particular,
when $\acal$ is locally finitely presentable, $M$ is finitary. But one
cannot expect that each finitary monad is determined by an algebraic
theory. A typical example is when $\acal$ is the category $\Vect_P$ of
vector spaces over a field $P$. Important binary operations are not
given by linear maps $V\times V\to V$ but by bilinear ones. So, linear
universal algebra deals with operations $V\otimes V\to V$. The
corresponding ``tensor algebraic theories'' were introduced by Mac Lane
\cite{ML1} under the names of PROP's and PACT's and led to the concept
of an \emph{operad}. We recommend \cite{L} to learn about linear
universal algebra. One can describe algebras over a monad
$M:\acal\to\acal$ for an arbitrary category $\acal$ by means of
``operations and equations'' (see \cite{Li}). However, arities of
operations are not natural numbers but objects of $\acal$. For
instance, by taking the category of posets, binary operations whose
arity is a two-element chain are defined only for pairs of comparable
elements.  It is just the special feature of $\Set$ that, besides
being cartesian closed, finite sets are coproducts of $1$ which makes
algebraic theories powerful enough to cover finitary monads.

Filtered categories are precisely categories $\dcal$ such that
colimits over $\dcal$ commute with finite limits in $\Set$ (see, e.g.,
\cite{ar}). There is also a characterization of filtered categories
independent of sets -- a category $\dcal$ is filtered if and only if
the diagonal functor $\triangle:\dcal\to\dcal^\ical$ is final for each
finite category $\ical$ (see \cite{GU}). It makes filtered colimits
belong more to the ``doctrine of finite limits'' than to that of
finite products.  It appears as the fact that algebras over a finitary
monad do not need to ``look like algebras''.  For example, the
category of torsion free abelian groups is the category of algebras
for the monad $M:\Ab\to\Ab$ on the category $\Ab$ of abelian groups
given by the reflection to torsion free ones.  This monad is finitary
because torsion free groups are closed under filtered colimits in
$\Ab$ but to be torsion free is not equationally definable. It would
be more appropriate to consider categories $\dcal$ such that colimits
over $\dcal$ commute with finite products. These categories are called
sifted and are characterized by the property that the diagonal
functor $\triangle:\dcal\to\dcal\times\dcal$ is final (see
\cite{adam-rosi:sifted}). 

\medskip\noindent Explicitely, a category $\dcal$ is \textbf{sifted}, if it is non-empty and for all objects $A,B\in\dcal$ the category $\mathsf{Cospan}_\dcal(A,B)$ of cospans is connected. Here, $\mathsf{Cospan}_\dcal(A,B)$ has as objects pairs of arrows $(A\to C, B\to C)$ and arrows $(A\stackrel{a}{\to} C, B\stackrel{b}{\to} C)\to(A\stackrel{a'}{\to} C', B\stackrel{b'}{\to} C')$ are given by arrows $f:C\to C'$ such that $f\circ a=a',f\circ b=b'$.  A category is connected if it is non-empty and cannot be decomposed into a disjoint union (coproduct) of two non-empty subcategories. 

\medskip\noindent Each filtered category is sifted but there
are sifted categories which are not filtered -- the most important
example are reflexive pairs (a parallel pair of morphisms $f,g$ is
reflexive if there is $t$ with $ft=gt=id$).  Another special feature
of sets is that any finitary functor $\Set\to\Set$ preserves sifted
colimits. But it is not true for $\Ab$ -- the torsion free monad above
does not preserve sifted colimits. The consequence is that torsion
free groups do not form a variety of universal algebras.

We can define algebras over an arbitrary functor $L:\acal\to\acal$; an
$L$-\emph{algebra} $(A,\alpha)$ is a pair consisting of an object $A$
and a morphism $\alpha:LA\to A$.  In the case when $L$ is a monad,
these $L$-algebras are more general than algebras over a monad $L$
because the latter have to satisfy some equations. Morphisms
$(A,\alpha)\to(A',\alpha')$ are morphisms $f:A\to A'$ such that
$f\circ\alpha=\alpha'\circ Lf$. The resulting category of $L$-algebras
is denoted by $\Alg(L)$. Like in the case of general equational
theories $(\Sigma,E)$, free $L$-algebras do not need to exists. If
they exist, the category $\Alg(L)$ is monadic over $\acal$ with
respect to the forgetful functor $\Alg(L)\to\acal$ sending an
$L$-algebra $(A,\alpha)$ to $A$. Given a functor $L:\acal\to\acal$,
algebras over $L^{\op}:\acal^{\op}\to\acal^{\op}$ are called
$L$-\emph{coalgebras}.
 
\section{Sifted Colimit Preserving Functors} 

The concept of a locally finitely presentable category stems from that
of a filtered colimit, i.e., belongs to the doctrine of finite
limits. Recall that an object $A$ of a category $\acal$ is finitely
presentable if its hom-functor $\hom(A,-):\acal\to\Set$ preserves
filtered colimits.  A category $\acal$ is locally finitely presentable
if it is cocomplete and has a set $\xcal$ of finitely presentable
objects such that each object of $\acal$ is a filtered colimit of
objects from $\xcal$.  By changing the doctrine from finite limits to
finite products, we replace filtered colimits by sifted ones. We say
that an object $A$ is \emph{strongly finitely presentable} if its
hom-functor $\hom(A,-):\acal\to\Set$ preserves sifted colimits. A
category $\acal$ is \emph{strongly locally finitely presentable} if it
is cocomplete and has a set $\xcal$ of strongly finitely presentable
objects such that each object of $\acal$ is a sifted colimit of
objects from $\xcal$. These categories were introduced in
\cite{adam-rosi:sifted} where it was shown that they are precisely
categories of algebras over many-sorted algebraic theories in
$\Set$. Recall that a \emph{many-sorted algebraic theory} $\tcal$ is a
small category with finite products and a $\tcal$-algebra in $\Set$ is
a functor $\tcal\to\Set$ preserving finite products. In particular,
each variety is strongly locally finitely presentable. (Note that our
usage differs from that of \cite{ar} where varieties are given by
many-sorted algebraic theories.) While finitely presentable objects in
$\Alg(\tcal)$ are algebras finitely presentable in a usual sense,
i.e., given by a finite set of generators subjected to a finite set of
equations, strongly finitely presentable algebras are precisely
retracts of finitely generated free algebras, i.e., finitely
presentable projective algebras. An important fact is that each
finitely presentable algebra is a reflexive coequalizer of strongly
finitely presentable ones (see
\cite[2.3.(2)]{adam-rosi:sifted}). Every strongly locally finitely
presentable category $\acal$ is locally finitely presentable and has,
up to isomorphism, only a set of strongly finitely presentable
objects. We will denote by $\acal_{\sfp}$ the corresponding full
subcategory of $\acal$. In the same way, $\acal_{\fp}$ denotes a
representative small full subcategory of finitely presentable objects.

Given a strongly locally finitely presentable category $\acal$ and a
category $\bcal$ having sifted colimits then a sifted colimit
preserving functor $H:\acal\to\bcal$ is fully determined by its values
on strongly finitely presentable objects. In fact, $\acal$ is a free
completion of $\acal_{\sfp}$, i.e., each functor
$\acal_{\sfp}\to\bcal$ extends to a functor $\acal\to\bcal$ (see
\cite{adam-rosi:sifted}). In particular, it applies to functors
$\acal\to\acal$. In analogy to \cite{ar}, Remark 2.75, we can prove
the following basic result.

\begin{thm}\label{thm:sfp} 
  Let $\acal$ be a strongly locally finitely presentable category and
  $L:\acal\to\acal$ preserve sifted colimits. Then $\Alg(L)$ is
  strongly locally finitely presentable.
\end{thm} 

\begin{proof} 
  Since $\acal$ is strongly locally finitely presentable and $L$
  preserves filtered colimits, $\Alg(L)$ is locally finitely
  presentable (see \cite[Remark 2.75]{ar}). Hence the forgetful
  functor $U:\Alg(L)\to\acal$ has a left adjoint $F$ (see \cite{ar},
  1.66). Thus $U$ is monadic, i.e., each $L$-algebra $X$ admits a
  regular epimorphism $e:FA\to X$ from some free $L$-algebra. Since
  $\acal$ is strongly locally finitely presentable, $A$ is a sifted
  colimit of strongly finitely presentable objects. Hence $FA$ is a
  sifted colimit of free algebras over strongly finitely presentable
  objects. Thus there is a regular epimorphism from a coproduct of
  free algebras over strongly finitely presentable objects to
  $FA$. Consequently, there is a strong epimorphism from such a
  coproduct to $X$, which means that free algebras over strongly
  finitely presentable objects form a strong generator of
  $\Alg(L)$. By \cite{CRV}, $\Alg(L)$ is strongly locally finitely
  presentable.
\end{proof}
 
\begin{thm}\label{thm:variety} 
  Let $\acal$ be a variety and $L:\acal\to\acal$ preserve sifted 
  colimits. Then $\Alg(L)$ is a variety. 
\end{thm} 
\begin{proof}
  Since $\acal$ is a variety, there is an object $S_1\in\acal_{\sfp}$
  such that each object of $\acal_{\sfp}$ is a retract of a finite
  coproduct of copies of $S_1$ ($S_1$ is a free algebra with one
  generator). Hence each object of $F(\acal_{\sfp})$ is a retract of
  finite coproducts of copies of $F(S_1)$. By \cite{CRV}, the closure
  of $F(\acal_{\sfp})$ under finite coproducts is the set $\xcal$ from
  the definition of a strongly locally finitely presentable category.
  Since a coproduct of retracts is a retract of coproducts, each object
  of $\xcal$ is a retract of finite coproducts of copies of $F(S_1)$.
  Hence $\Alg(L)$ is a variety.
\end{proof}

\noindent In some simple but important varieties like sets or vector
spaces, every finitely presentable algebra is projective. % As a
 
\begin{prop} 
  Let $\acal$ be a variety such that every finitely presentable
  algebra is projective. Then any functor $L:\acal\to\acal$ preserving
  filtered colimits preserves sifted colimits.
\end{prop} 
 
\begin{proof}
  Let $L:\mathcal A\to\mathcal A$ preserve filtered colimits.  Then
  $L$ is uniquely determined by its restriction $L_0$ to $\mathcal
  A_{\fp}$. Since $\mathcal A_{\fp}=\mathcal
  A_{\sfp}$, there is a unique extension $L':\mathcal
  A\to\mathcal A$ of $L_0$ preserving sifted colimits. Since $L=L'$,
  $L$ preserves sifted colimits.
\end{proof} 

\noindent The previous proposition can be extended to Boolean
algebras. In fact, the trivial Boolean algebra $\onebb$ is the only
finitely presentable Boolean algebra that is not projective. $\onebb$
is the reflexive coequalizer
\begin{equation}\label{equ:refl-coequ} 
\xymatrix{F1\ar@<.6ex>[r]^{i}  
             \ar@<-.6ex>[r]_{o}  
  & F0 \ar[r]^{s}  
  &  \onebb  
} 
\end{equation} 
where $F$ is the left adjoint to the forgetful functor $\BA\to\Set$, 
$i$ maps the generator to the top, and $o$ maps the generator to the 
bottom.  If $L:\BA\to\BA$ preserves filtered colimits and the above 
coequalizer, then $L$ preserves sifted colimits. 
 
\begin{prop}\label{prop:sifted=filtered-BA} 
  For any filtered colimit preserving functor $L:\BA\to\BA$ there is a
  sifted colimit preserving functor $L':\BA\to\BA$ such that $L$ and
  $L'$ are isomorphic when restricted to the full subcategory of $\BA$
  without $\onebb$. Moreover, $\Alg(L)=\Alg(L')$.
\end{prop} 

\begin{proof} 
  If $A\neq\onebb$ then there is no arrow $\onebb\to A$. Thus $A$ is a
  filtered colimit of objects from $\BA_{\sfp}$. Let $L_0$ be the
  restriction of $L$ to $\BA_{\sfp}$ and $L':\BA\to\BA$ be the sifted
  colimit preserving extension of $L_0$. Then $L'$ is isomorphic to
  $L$ on the full subcategory of $\BA$ without $\onebb$. The rest is
  evident.
\end{proof}

\noindent The proposition shows that as far as we are concerned with
algebras over $\BA$, we can assume any finitary functor to preserve
sifted colimits. It also gives a category theoretic reason for
sometimes restricting attention to non-trivial Boolean algebras.

\section{Presenting Functors on Varieties}\label{sec:present-fun}

Given a strongly locally finitely presentable category $\acal$, this
section shows that a functor $L:\acal\to\acal$ has a finitary
presentation by operations and equations iff $L$ preserves sifted
colimits. We start by investigating the category $\scal(\acal)$ of
all functors $\acal\to\acal$ preserving sifted colimits. Morphisms are
natural transformations.  It is a legitimate category because it is
equivalent to $\acal^{\acal_{\sfp}}$.

\begin{prop}\label{prop4.1}
  Let $\acal$ be a strongly locally finitely presentable
  category. Then $\scal(\acal)$ is strongly locally finitely
  presentable.
\end{prop}
\begin{proof}
  There is a small category $\tcal$ with finite products such that
  $\acal$ is equivalent to the category of all functors $\tcal\to\Set$
  preserving finite products. Since
$$
(\Set^\tcal)^{\acal_{\sfp}}\cong\Set^{\tcal\times\acal_{\sfp}},
$$
the category $\acal^{\acal_{\sfp}}$ is equivalent to the category of all functors
$$
\tcal\times\acal_{\sfp}\to\Set
$$
sending cones
$((p_i,id_S):(T_1,S)\times\dots\times(T_n,S)\to(T_i,S))_{i=1}^n$ to
products; here $(p_i:T_1\times\dots\times T_n\to T_i)_{i=1}^n$ is a
finite product in $\tcal$. Hence $\acal^{\acal_{\sfp}}$ is equivalent
to the category of models of a finite product sketch and thus it is
strongly locally finitely presentable (see \cite{ar}, 3.17). Hence
$\scal(\acal)$ is strongly locally finitely presentable.
\end{proof}

A functor $H:\acal\to\bcal$ between strongly locally finitely
presentable categories is called \emph{algebraically exact} provided
that it preserves limits and sifted colimits. Then $H$ has a left
adjoint and the reason for this terminology is that such functors
dually correspond to morphisms of many-sorted algebraic theories (see
\cite{ALR}). In more detail, given a finite product preserving functor
$M:\tcal_1\to\tcal_2$ between categories with finite products, then
the corresponding algebraically exact functor $H$ sends a
$\tcal_2$-algebra $A$ to the composition $AM$.

Let $\acal$ be a variety, $U:\acal\to\Set$ the forgetful functor and
$F$ its left adjoint. We get functors
$$
\Psi:\scal(\acal)\to\scal(\Set)
$$
and
$$
\Phi:\scal(\Set)\to\scal(\acal)
$$
by means of $\Psi(L)=ULF$ and $\Phi(G)=FGU$. The definition is correct
because $U$ preserves sifted colimits and $F$ preserves all colimits.

\begin{prop}\label{prop4.2}
  Let $\acal$ be a variety. Then the functor $\Psi$ is algebraically
  exact and $\Phi$ is its left adjoint.
\end{prop}
\begin{proof}
  The functor $\Psi$ is equivalent to the composition
$$
\acal^{\acal_{\sfp}}\xrightarrow{\ U^{\acal_{\sfp}}\ }
\Set^{\acal_{\sfp}}\xrightarrow{\ \Set^{F_{\sfp}}\ }
\Set^{\Set_{\sfp}}
$$
where $F_{\sfp}$ denotes the restriction of $F$ to
strongly finitely presentable objects. Clearly, $\Psi$ preserves
limits and sifted colimits. It
remains to show that $\Phi$ is left adjoint to $\Psi$. This left
adjoint is equivalent to the composition
$$
\Set^{\Set_{\sfp}}\xrightarrow{\ \Phi_1\ } \Set^{\acal_{\sfp}}\xrightarrow{\
F^{\acal_{\sfp}}\ } \acal^{\acal_{\sfp}}
$$
where $\Phi_1$ is left adjoint to $\Set^{F_{\sfp}}$. Since
$\Set_{\sfp}$ is the category $\Fin$ of finite sets and each functor
$\Fin\to\Set$ is a colimit of hom-functors $\hom(k,-)$ where $k$ is a
finite cardinal, it suffices to show that the left adjoint to $\Psi$
coincides with $\Phi$ on hom-functors $\hom(k,-)$. But it follows from
$$
\Phi_1(\hom(k,-))=\hom(Fk,-)\cong\hom(k,U-)=\hom(k,-)U.
$$
\end{proof}

\begin{rem}\label{rmk:S(A)algebraic}
  The adjoint transpose of $\tau:G\to ULF$ is $FGU\stackrel{F\tau
    U}{\too} FULFU\stackrel{\eps
    LFU}{\too}LFU\stackrel{L\eps}{\too}L$, which we can also write as
  $(GU\stackrel{\tau U}{\too} ULFU\stackrel{UL\eps}{\too}UL)^\dagger$
  where $\dagger$ denotes the adjoint transpose wrt $F\adj U$.
\end{rem}
 
As it is well-known, presentations can be obtained as follows.

\begin{prop}\label{prop:pres-exist}
  Let $H:\acal\to\bcal$ be an algebraically exact functor between
  strongly locally finitely presentable categories with a left adjoint
  $F$ such that the counit $\varepsilon$ is a pointwise regular
  epimorphism. Then each object $A$ in $\acal$ has a presentation as a
  coequalizer
\begin{equation}\label{equ:present-alg-prop} 
 \xymatrix{ 
   FR \ar@<-0.6ex>[r]_{r_2^\sharp} \ar@<0.6ex>[r]^{r_1^\sharp} & 
    FB \ar[r]^{e} & A 
 } 
\end{equation} 
where $B$ is an object of $\bcal$, $R$ a subobject of $HFB\times HFB$
and $r_1^\sharp,r_2^\sharp$ are adjoint transposes of the projections
$r_1,r_2:RB\to HFB$. 
\end{prop}

\begin{proof}
  It suffices to take a regular epimorphism $e:FB\to A$ (such as the
  counit $\eps_A$), its kernel pair $e_1,e_2:C\to FB$ and to put
  $r_i=He_i$, $i=1,2$. The claim then follows because in strongly
  locally finitely presentable categories any regular epi is the
  coequalizer of its kernel pair and because, by assumption, the
  counit $FHC\to C$ is regular epi.
\end{proof}

 We will need the following modification.

\begin{lem}\label{le4.3}\label{lem:pres-exist}
  Let $H_1:\acal\to\bcal$ and $H_2:\bcal\to\ccal$ be algebraically
  exact functors between strongly locally finitely presentable
  categories with left adjoints $F_1$ and $F_2$, respectively, such
  that both counits $\varepsilon_1,\varepsilon_2$ are pointwise
  regular epimorphisms. Then each object $A$ in $\acal$ has a
  presentation as a coequalizer
\begin{equation}\label{equ:present-alg-lemma} 
 \xymatrix{ 
   F_1R \ar@<-0.6ex>[r]_{r_2^\sharp} \ar@<0.6ex>[r]^{r_1^\sharp} & 
     F_1F_2C \ar[r]^{e} & A 
 } 
\end{equation} 
where $C$ is an object of $\ccal$, $R$ a subobject of $H_1F_1F_2C\times
H_1F_1F_2C$ and $r_1^\sharp,r_2^\sharp$ are adjoint transposes of the
projections $r_1,r_2:R\to H_1F_1F_2C$.
\end{lem}
\begin{proof}
  Consider the composition $e=e_1F_1(e_2)$, where $e_1:F_1B\to A$ and
  $e_2:F_2C\to B$ are regular epis. Note that $F_1$ preserves regular
  epis since it is a left-adjoint and that in a strongly locally
  finitely presentable category regular epis are closed under
  composition since, in many-sorted varieties, regular epis are
  precisely sort-wise surjective homomorphisms. Now we follow
  Prop~\ref{prop:pres-exist} by taking $r_i=H_1p_i$ and $p_1,p_2$ to
  be the kernel pair of $e$.
\end{proof}

Let $\acal$ be a variety and $\Psi:\scal(\acal)\to\scal(\Set)$ the
algebraically exact functor from \ref{prop4.2}.  Since $\scal(\Set)$
is equivalent to $\Set^\Fin$, it is an $\mathbb{N}$-sorted variety where
$\mathbb{N}$ is the set of non-negative integers. Hence there is another
algebraically exact functor $H_2:\scal(\Set)\to\Set^\mathbb{N}$.  Its left
adjoint $F_2$ sends an $\omega$-sorted set $(G_k)_{k<\omega}$ to the
functor $G:\Set\to\Set$ given as
$$
GX=\coprod_{k<\omega} G_k\times X^k.
$$
We are going to show that Lemma~\ref{le4.3} leads to the presentation of a
functor $L:\acal\to\acal$ as in \cite{bons-kurz:fossacs06}.
 
\begin{definition}[\cite{bons-kurz:fossacs06}]
  \label{def:present-fun}
  A finitary presentation by operations and equations of a functor is
  a pair $\langle G,E\rangle$ where $G:\Set\to\Set$, $GX=\coprod_{k<
    \omega} G_k\times X^k$ and $E=(E_V)_{V\in\omega}$, $E_V\subseteq
  (UFGUFV)^2$. The functor $L$ presented by $\langle G,E\rangle$ is
  the multiple coequalizer
\begin{equation}\label{equ:present-fun} 
\xymatrix{ FE_V\ar@<0.6ex>[r]^{\pi_1^\dagger\ \ \ \ }  
             \ar@<-0.6ex>[r]_{\pi_2^\dagger\ \ \ \ }  
  & FGUFV\ar[r]^{\ \ FGUv}  
  &  FG UA \ar[r]^{\ \ \ q_A} & LA } 
\end{equation} 
where $\pi_{V,i}^\dagger$ are the adjoint transposes of the projections
$E_V\to UFGUFV$; $V$ ranges over finite cardinals and $v$ over
morphisms (valuations of variables) $FV\to A$.
\end{definition} 
 
\begin{example}%[modal algebras]
\label{exle:mod-alg} 
A modal algebra, or Boolean algebra with operator (BAO), consists of a
Boolean algebra $A$ and a meet-preserving operation $A\to
A$. Equivalently, a BAO is an algebra for the functor
$L:\BA\to\BA$, where $L A$ is defined by generators $\Box a$, $a\in
A$, and relations $\Box\top=\top$, $\Box(a\et a')= \Box a \et \Box
a'$. That is, in the notation of the definition, $GX=X$,
$E_V=\emptyset$ for $V\not=2$, $E_2=\{\Box\top=\top, \Box(v_0\et
v_1)=\Box v_0 \et \Box v_1\}$.
\end{example} 
 
\noindent In \cite{pattinson:cml-j,kkp:cmcs04,schroeder:fossacs05}
`modal axioms of rank 1' play a prominent role. These are exactly
those which, considered as equations, are of the form $E_V\subseteq
(UFGUFV)^2$.

\begin{definition}[rank 1]\label{def:rank1}
  Let $\acal$ be a variety with equational presentation
  $\langle\Sigma_\acal,E_\acal\rangle$. Consider a collection $\Sigma$
  of additional operation symbols and a set $E$ of equations in
  variables $V$ over the combined signature $\Sigma_\acal+\Sigma$. We
  say that the equations $E$ are of \emph{rank 1} if every variable is
  under the scope of precisely one operation symbol from $\Sigma$, or
  more formally, $E\subseteq (UFGUFV)^2$ where $G:\Set\to\Set$ is the
  endofunctor associated with $\Sigma$ and $F\adj U$ is the adjunction
  associated with $U:\acal\to\Set$.
\end{definition}
 
\begin{rem}\hfill
\begin{enumerate}[(1)]
\item The generators appear as a functor $G$. This expresses that the
  same generators (the $\Box$ in the example above) are used for all
  $LA$ where $A$ ranges over $\BA$.  Similarly, the coequalizer
  (\ref{equ:present-fun}) is expressed using equations in variables
  $V$, that is, the same relations are used for all $LA$. In
  $E_V\subseteq (UFGUFV)^2$ the inner $UF$ allows for the conjunction
  in $\Box(v_0\et v_1)$ whereas the outer $UF$ allows for the
  conjunction in $\Box v_0 \et \Box v_1$. 
\item It is often useful to analyse algebras over algebras using
  monads and distributive laws. But in our situation, as will be clear
  from the following proof, $UFGUFV$ arises not from applying the
  monad $UF$ to $V$ and then to $GUFV$, but from applying to $G$ the
  functor $\Phi=F-U$ followed by $\Psi=U-F$.
\end{enumerate}
\end{rem} 

\begin{thm}\label{thm:present-sift} 
  An endofunctor on a variety has a finitary presentation by
  operations and equations if and only if it preserves sifted
  colimits.
\end{thm} 
\begin{proof}
  For `if', let $\acal$ be a variety, $U:\acal\to\Set$ the forgetful
  functor, $F$ its left adjoint, and let $L:\acal\to\acal$ preserve
  sifted colimits.
  % Since strongly finitely presentable objects
  % in $\acal$ are precisely retracts of finitely generated free
  % algebras, the category $\scal(\acal)$ is equivalent to
  % $\acal^{F(\Fin)}$. A functor $L:\acal\to\acal$ preserving sifted
  % colimits corresponds to its restriction $L':F(\Fin)\to\acal$.
  We apply Lemma \ref{le4.3} to $H_1=\Psi$ and
  $H_2:\scal(\Set)\to\Set^\omega$ described above. It presents $L$ as a
  coequalizer
\begin{equation}\label{equ:present-alg-1} 
 \xymatrix{ 
   \Phi R \ar@<-0.6ex>[r]_{r_2^\sharp} \ar@<0.6ex>[r]^{r_1^\sharp} & \Phi G \ar[r]^{e} & L
 } 
\end{equation}
where $(G_k)_{k<\omega}$ is an $\omega$-sorted set and
$G=\coprod_{k<\omega} G_k\times (-)^k$.  It yields a finitary
presentation of $L$ with $E_V$ given by $RV$. To verify this in detail
consider
\begin{equation}\label{equ:present-alg-2} 
  \xymatrix{ 
    FRUA \ar@<-0.6ex>[r] \ar@<0.6ex>[r]^{r_i^\sharp A} & 
    FGUA \ar@{..>}^{q_A}[r] & LA \\
    FRUFV \ar@<-0.6ex>[r]\ar[u]^{FRUv} \ar@<0.6ex>[r]^{r_i^\sharp FV } \ar[d]_{FUP\eps FV} & 
    FGUFV \ar@{..>}[u]_{FGUv} &\\
    FRV \ar@{..>}@<-0.6ex>[ur]_{\pi_{V,i}^\dagger} \ar@{..>}@<0.6ex>[ur] &  &      
  }
\end{equation}
where the upper row is Diagram~\ref{equ:present-alg-1}, and the dotted
arrows are Diagram~\ref{equ:present-fun} ($\dagger$ denotes adjoint
transpose wrt $F\adj U$ and $\sharp$ wrt $\Phi\adj \Psi$). To check
that the two triangles commute, recall from Lemma~\ref{lem:pres-exist}
that $R=UPF$ arises from the kernel $p_1,p_2:P\to FGU$ of $FGU\to
L$. We put $r_i=Up_iF$ and $\pi_{V,i}=Up_iFV$. We have $r^\sharp_i FV=
FGU\eps FV\circ \eps FGUFUFV \circ FUpFUFV$
(Remark~\ref{rmk:S(A)algebraic}) and $\pi_{V,i}^\dagger = \eps
FGUFV\circ FUpFV$. Commutativity of the triangles now follows from the
naturality of $FUp$ and $\eps FGU$.
To finish the argument, recall that the upper row is a
coequalizer. Because $A$ is a sifted colimit of finitely generated
free algebras $FV$ and because $FRU$ preserves sifted colimits, it
follows that $q_A$ is the multiple coequalizer obtained from the
middle row.  Since $FUP\eps FV$ is epi ($\eps F$ is split by $F\eta$)
the dotted arrows also form a multiple coequalizer.

Conversely, every functor $L:\acal\to\acal$ with a presentation
preserves sifted colimits because such a functor is a coequalizer of
two natural transformations between functors preserving sifted
colimits.  In more detail, suppose that $L$ has a presentation as in
(\ref{equ:present-fun}).  Let $c_i:A_i\to A$ be a sifted colimit.  We
have to show that $Lc_i$ is a sifted colimit. Given a cocone
$d_i:LA_i\to L'$ we have to show that there is a unique $k$ as
depicted in
\[ 
\xymatrix@C=20pt{  
FE_V\ar@<.6ex>[r]^{\pi_1^\sharp\ \ \ \ }  
             \ar@<-.6ex>[r]_{\pi_2^\sharp\ \ \ \ }  
  & FGUFV\ar[r]^{} \ar[dr]_{\ \ FG Uv^\sharp}  
    &  FGUA_i \ar[r]^{\ \ q_{A_i}} \ar[d]^{FGUc_i} 
      & LA_i \ar[d]^{Lc_i} \ar[dr]^{d_i} 
        & \\ 
  &  
    & FGUA \ar[r]_{\ \ q_A} \ar `d[r] `/4pt[rr]_{h} [rr] 
      & LA \ar@{..>}[r]_{k} 
        & L'  
} 
\] 
$U$ preserves sifted colimits because $\acal$ is a variety, $G$ 
preserves sifted colimits because they commute with finite products, 
and $F$ preserves all colimits.  Therefore we have $h$ with 
$d_i\circ q_{A_i}=h\circ FGUc_i$. Then $k$ is obtained from the joint 
coequalizer $q_A$ once we show that $h\circ FGUv^\sharp\circ 
\pi_1^\sharp = h\circ FGUv^\sharp\circ \pi_2^\sharp$ for all $v:V\to 
UA$. For this consider $v:V\to UA$.  Since $\hom(V,-)$ preserves 
sifted colimits ($V$ is finite) and $U$ preserves sifted colimits, 
there is some $A_j$ and some $w:V\to UA_j$ such that $v=Uc_j\circ v'$. 
It follows $v^\sharp= c_j\circ w^\sharp$, hence 
$FGUv^\sharp=FGUc_j\circ FGU w^\sharp$. 
\end{proof}

\begin{rem}[correspondence between functors and presentations]
\label{rmk:pres-from-functor}
We summarise the constructions of the proof for future reference.
\begin{enumerate}[(1)]
\item Let $L$ be a sifted colimit preserving functor. Then we obtain a
  finitary presentation $\langle\Sigma,E\rangle$ as follows.
  Given $L$ we find a suitable $G$ as $GX=\coprod_{n<\omega}ULFn\times
  X^n$, with $GX\to ULFX$ given by $(\sigma\in ULFn,v:n\to X)\mapsto
  ULF(v)(\sigma)$. The quotient $q:FGUA\to LA$ is then given by (see
  Remark~\ref{rmk:S(A)algebraic}) the adjoint transpose of $GUA\to
  ULA$, mapping $(\sigma,v:n\to UA)\mapsto UL(v^\dagger)(\sigma)$. To
  summarise, the set of $n$-ary operations of $\Sigma$ is $ULFn$ and
  the set of equations in $n$ variables is the kernel of $q:FGUFn\to
  LFn$.
\item Conversely, every presentation $\langle\Sigma,E\rangle$ defines
  a functor as in Definition~\ref{def:present-fun}.
\end{enumerate}
\end{rem}

Given a variety $\acal$ and a functor $L:\acal\to\acal$ preserving
sifted colimits, we know that $\Alg(L)$ is a variety (see
\ref{thm:variety}). The main point is that one obtains a presentation
of $\Alg(L)$ from a presentation of $\acal$ and from a presentation of
$L$.

\begin{thm}[\cite{bons-kurz:fossacs06}]\label{thm:present-fun}
 Let $\acal\iso\Alg(\Sigma_\acal, E_\acal)$ be a variety and $\langle 
  \Sigma_L, E_L\rangle$ a finitary presentation of $L:\acal\to\acal$. Then 
  $\Alg(\Sigma_\acal + \Sigma_L, E_\acal+E_L)$ is isomorphic to 
  $\Alg(L)$, where equations in $E_\acal$ and $E_L$ are understood as 
  equations over $\Sigma_\acal + \Sigma_L$. 
\end{thm} 
 
\begin{rem} \label{thm:present-fun:rmk} The theorem shows that, for
  a functor $L$ determined by its action on the finitely generated
  free algebras of a variety $\acal$, the notion of $L$-algebra is a
  special case of the universal algebraic notion of algebra for
  operations and equations.
\begin{enumerate}[(1)]
\item In detail, given an algebra $\alpha:LA\to A$, we define a
  $(\Sigma_\acal+\Sigma_L)$-algebra structure $\mathsf{A}$ on $A$ via
  $\sigma^\mathsf{A}(a_1,\ldots a_n)=\alpha(\sigma(a_1,\ldots a_n))$
  for all $n$-ary operations $\sigma\in\Sigma_L$ and all $(a_1,\ldots
  a_n)\in A^n$. 
\item Conversely, given $\mathsf{A}\in\Alg(\Sigma_\acal + \Sigma_L,
  E_\acal+E_L)$, we can define $\alpha:LA\to A$ by
  $\alpha(\sigma(a_1,\ldots a_n))=\sigma^{\mathsf{A}}(a_1,\ldots a_n)$
  for all $n$-ary operations $\sigma\in\Sigma_L$ and all $(a_1,\ldots
  a_n)\in A^n$. Since $LA$ is freely constructed from generators this
  determines $\alpha$ on all of $LA$.  
\item The logical significance of the theorem is that it ensures that the
  Lindenbaum algebra for the signature $\Sigma_\acal + \Sigma_L$ and
  the equations $E_\acal+E_L$ is the initial $L$-algebra.
\end{enumerate}
\end{rem}
 
\noindent A consequence of Theorem~\ref{thm:present-sift} is the
immediate proof of the fact that functors having a presentation are
closed under composition (see \cite{bons-kurz:fossacs06}). Also, they
preserve surjections as regular epis are reflexive coequalizers of
their kernel pairs and hence sifted colimits.

A previous draft of this paper posted by the first author contained a
wrong statement about the preservation of injections. Although it is
true that a sifted colimit preserving functor on $\BA$ preserves
injections, this does not in general extend to other varieties. For
example, take the category of semi-groups and the functor $L$ given by
the presentation consisting of a unary operator $\Box$ and equations
$\Box(a\circ b)= (\Box a) \circ (\Box b)$ and $\Box(a\circ
b)=\Box(a'\circ b')$. The first equation on its own would just specify
that $LA$ is isomorphic to $A$. But the second equations means that
all elements of $A$ that can be decomposed must be identified. Now
consider the non-negative integers $\Nbb$ with addition as the
semi-group operation.  Then $L$ does not preserve the injectivity of
inclusion $ \Nbb\setminus\{0\}\to\Nbb$ since $L(\Nbb\setminus\{0\})$
has two elements whereas $L\Nbb$ has only one element. More
importantly, Rob Myers \cite{myers:phd} has an example of an
equationally presented functor on distributive lattices that does not
preserve injections.

\section{Introduction to Part II: Duality of Algebras and Coalgebras}
This second part of the paper proves a representation theorem for
functor-algebras based on Stone's representation theorem for Boolean
algebras and extending the J\'onsson-Tarski theorem for modal algebras
(Boolean algebras with operators).

\medskip\noindent\textbf{Stone's representation theorem} for Boolean
algebras shows that any Boolean algebra $A$ can be represented as an
algebra of subsets where the Boolean operations are interpreted
set-theoretically (conjunction as intersection, etc). For the proof,
one identifies a functor $\Sigma:\BA\to\Set\op$ and a $\BA$-morphism
\begin{equation}\label{equ:Stone}
  \iota_A: A\to \Pi\Sigma A
\end{equation}
into the powerset $\Pi(\Sigma A)$ and then shows that $\iota_A:A\to
\Pi\Sigma A$ is injective, exhibiting $A$ as isomorphic to a
subalgebra of a powerset.

\medskip\noindent Analysing this situation from a categorical point of
view one finds that
\begin{equation}
  \xymatrix{ 
    {\BA} \ar@/_/[rr]_{\Sigma}  & &  {\Set\op}\ar@/_/[ll]_{\Pi}  
    \\ 
    \ar[u] \BA_\omega\ar@/_/[rr]  & & \ar[u] \Set_\omega\op\ar@/_/[ll]
  } 
\end{equation}
(i) the category $\BA_\omega$ of finite Boolean algebras is dually
equivalent to the category $\Set_\omega$ of finite sets, (ii) $\BA$
and $\Set$ are the completion under filtered colimits, or
ind-completion, of $\BA_\omega$ and $\Set_\omega$, (iii) the two
functors $\Pi$ and $\Sigma$ appearing in Stone's representation
theorem arise from lifting the duality between $\BA_\omega$ and
$\Set_\omega$ to the completions $\BA$ and $\Set$. In such a
situation, $\Sigma$ is left-adjoint to $\Pi$ and the representation
morphism (\ref{equ:Stone}) is the unit of the adjunction.

\medskip\noindent Abstracting from the particularities of finite
Boolean algebras and finite sets leads us to replace $\BA_\omega$ by
an arbitrary small, finitely complete and co-complete category
$\ccal$ and to identify $\BA$ and $\Set\op$ as the so-called ind- and
pro-completions of $\BA_\omega$:
\begin{equation}\label{equ:IndPro} 
\xymatrix{ 
{\Ind} \ccal%\ar@(dl,ul)[]^{G \ } 
  \ar@/_/[rr]_{\Sigma}  & &  
{\Pro\ccal}\ar@/_/[ll]_{\Pi} %\ar@(dr,ur)[]_{\ H}  
\\ 
& \ar[ul]^{\hat{(-)}} \ccal  \ar[ur]_{\bar{(-)}} & 
} 
\end{equation} 
We summarise what we need to know about the diagram above, details can
be found in Johnstone~\cite[VI.1]{johnstone:stone-spaces}. In the
diagram $\hat{(-)}:\ccal\to\Ind\ccal$ is the completion of $\ccal$
under filtered colimits and $\bar{(-)}:\ccal\to\Pro\ccal$ under
cofiltered limits. Since $\ccal$ is finitely cocomplete, we have that
$\Ind\ccal$ is cocomplete. $\Ind\ccal$ is also complete. Dually
$\Pro\ccal$ is complete and cocomplete. Then $\Sigma$ is defined to be
the unique extension of $\bar{(-)}$ along $\hat{(-)}$ preserving all
colimits, and $\Pi$ is the unique extension of $\hat{(-)}$ along
$\bar{(-)}$. The functors $\Sigma$ and $\Pi$, also known as
Kan-extensions, restrict to isomorphisms on $\ccal$, that is,
\begin{equation}\label{equ:SigmaPi} 
\Sigma \hat C \cong \bar C \quad \quad \Pi\bar C \cong \hat C 
\end{equation} 
 
\begin{example}\label{exle:IndPro}
  Let $\acal$ be a variety and $\acal_\fp$ be the full subcategory of
  finitely presentable algebras. $\acal_\fp$ is closed under finite
  colimits. If $\acal$ is a \emph{locally finite variety}, that is, if
  finitely generated free algebras are finite, then $\acal_\fp$ is
  also closed under finite limits and $\acal_\fp$ is an example of a
  finitely complete and finitely cocomplete category $\ccal$. In
  particular we have the following instances of our general situation.
\begin{enumerate}[(1)]
\item $\ccal=\BA_\omega$ (finite Boolean algebras = finitely
  presentable Boolean algebras), $\Ind\ccal=\BA$, $\Pro\ccal=\Set\op$.
  $\Sigma A$ is the set of ultrafilters over $A$ and $\Pi$ is
  (contravariant) powerset.
\item $\ccal=\DL_\omega$ (finite distributive lattices = finitely
  presentable distributive lattices), $\Ind\ccal=\DL$,
  $\Pro\ccal=\Poset\op$. $\Sigma A$ is the set of prime filters over
  $A$ and $\Pi$ gives the set of downsets.
\end{enumerate}
\end{example} 
 
\noindent The following is a well-known fact.
 
\begin{prop}\label{prop:IndProAdj} 
$\Sigma$ is left adjoint to $\Pi$.  
\end{prop} 

\begin{proof}
  Any left Kan extension such as $\Sigma$ has a right adjoint $R$
  given by $RY=\Pro\ccal(\bar{(-)},Y)$. As $\Pi$ is the unique
  extension of $\hat{(-)}$ preserving cofiltered limits, the
  proposition follows from $\Pi$ and $R$ agreeing on $\ccal$, which,
  in turn, is a consequence of the Yoneda lemma. \end{proof}

\medskip\noindent\textbf{The J\'onsson-Tarski theorem}
\cite{jons-tars:bao1} extends Stone's theorem to modal algebras, or,
Boolean algebras with operators. For example, the BAO of
Example~\ref{exle:mod-alg} is a Boolean algebra with one unary
operation $\Box$ to interpret the modal operator. Given a BAO $A$, we
can associate with it a dual Kripke frame $(X,R_\Box)$, where $X$ is
$\Sigma A$ as above (Example~\ref{exle:IndPro}). Describing $\Sigma A$
as the set of ultrafilters of $A$, the relation $R_\Box$ is given explicitely by
\begin{equation}\label{equ:R-Box}
  xR_\Box y\
  \Leftrightarrow \ \forall a\in A\;.\;\Box a \in x \Rightarrow y\in
  a,
\end{equation}
  that is, $R_\Box$ is the largest relation such that, in logical
  notation, $(xR_\Box y \;\&\; x\Vdash\Box a)\ \Rightarrow \ y\Vdash
  a$. Conversely, to any Kripke frame $(X,R)$, one can associate the
  so-called complex algebra $(\Pi X, \Box_R)$ where $\Pi$ is powerset
  and $\Box_R\; a=\{x\in X\mid \forall y\ .\ xRy \Rightarrow y\in
  a\}$. The J\'onsson-Tarski theorem then states that the representation
  map (\ref{equ:Stone}) is not only a Boolean algebra homomorphism but
  also BAO-morphism
\begin{equation}\label{equ:jons-tars}
(A,\Box)\too (\Pi\Sigma A, \Box_{R_\Box}).
\end{equation} 
In our category theoretic reconstruction, the additional operator
$\Box$ corresponds to a functor $H$ and a BAO to an algebra $HA\to
A$. The relational structure corresponds to an algebra $KX\to X$ for a
functor $K$ on $\Pro\ccal$. (For the purpose of this part, it is
notationally easier to work with algebras on $\Pro\ccal$ rather than with coalgebras on
$(\Pro\ccal)\op$.)

\medskip\noindent It is interesting to note that the co-unit
$\Sigma\Pi X\to X$ of the adjunction does not lift to a morphism
between Kripke frames (written now in
$(\Pro\ccal)\op$) $$(X,R)\too(\Sigma\Pi X, R_{\Box_R})$$ as it is only
a graph homomorphism lacking the backward condition of a Kripke frame
morphism. This makes the J\'onsson-Tarski representation theorem
particularly interesting. In our category theoretic reconstruction, it
means that (i) $\Sigma$ does not lift to a \emph{functor}
$\Alg(H)\to\Alg(K)$ in (\ref{equ:tildeSigma}) and (ii)
$h:K\Sigma\to\Sigma H$ is not required to be natural.

\medskip\noindent In (\ref{equ:jons-tars}), $(\Pi\Sigma A,
\Box_{R_\Box})$ is known as the \textbf{canonical extension} of
$(A,\Box)$. The theory of canonical extensions, also going back to
J\'onsson-Tarski \cite{jons-tars:bao1} (but see eg \cite{venema:ac}
for a more recent overview), studies the following question: Suppose
that the BAO $A$ satisfies some equation $e$, does its canonical
extension $\Pi\Sigma A$ then satisfy $e$. Investigations of this kind
are beyond the scope of the paper.  Our generalisation of the
J\'onsson-Tarski theorem only concerns algebras for a
\emph{functor}. In terms of additional equations $e$, this means that
our result only shows that equations of rank 1
(Definition~\ref{def:rank1}) are preserved under canonical
extensions. Of course, in the case of Kripke frames this (and much
more) is already known, but the point of this paper is to generalise
to other functors $T$ (or $K$, as they are called in Part II).

\section{Representing Algebras on Ind-Completions}\label{sec:IndPro} 
\noindent 
We want to present algebras over $\Ind\ccal$ by coalgebras over
$\Pro\ccal\op$, or equivalently, by algebras over
$\Pro\ccal$. Therefore, with $\ccal, \Sigma, \Pi$ as in
Diagram~\ref{equ:IndPro}, we consider now $H:\Ind\ccal\to\Ind\ccal$
and $K:\Pro\ccal\to\Pro\ccal$
\begin{equation}\label{equ:IndProAlg} 
\xymatrix{ 
{\Ind\ccal} \ar@(dl,ul)[]^{H \ } 
  \ar@/_/[rr]_{\Sigma}  & &  
{\Pro\ccal}\ar@/_/[ll]_{\Pi} \ar@(dr,ur)[]_{\ K}  
} 
\end{equation} 
We write $\iota:\Id\to\Pi\Sigma$ and $\eps:\Sigma\Pi\to\Id$ for the
unit and co-unit of the adjunction and note that $\iota_{\hat C}$ and
$\eps_{\bar C}$ are isomorphisms for $C\in\ccal$. 

\medskip\noindent We say that \emph{$H$ is determined by $K$ on
  $\ccal$} if there is an isomorphism
\begin{align}
  \kappa_C\ \ : \ \ H\hat C \stackrel{\cong}{\longrightarrow} \Pi K \Sigma
  \hat C\label{equ:HK-H}
\end{align} 
natural in $C\in\ccal$; we say that \emph{$K$ restricts to $\ccal$}
if the counit $\eps $ is an iso on $K\bar C$
\begin{equation}\label{equ:HK-restrict} 
  \eps_{K\bar C}\ \ : \ \ \Sigma\Pi K \bar C \stackrel{\cong}{\longrightarrow}  K\bar C 
\end{equation} 
Together, \eqref{equ:HK-H} and \eqref{equ:HK-restrict} give an
isormorphism
\begin{align}
(\eps_{K\Sigma\hat C}\circ\Sigma\kappa_C)^{-1}\ \ :\ \  K \Sigma\hat C \stackrel{\cong}{\longrightarrow} \Sigma H \hat C \label{equ:HK-K}
\end{align} 
Recalling \eqref{equ:SigmaPi}, we remark that \eqref{equ:HK-H} and
\eqref{equ:HK-restrict} can be written more symmetrically saying that
$H$ and $K$ agree on $\ccal$:
\begin{equation}\label{eq:HKagree}
 H\hat C  \cong \Pi K \bar C \quad\quad\quad\quad 
K\bar C \cong \Sigma H \hat C
\end{equation} 
Conversely, these \eqref{eq:HKagree} implies \eqref{equ:HK-H} and
\eqref{equ:HK-restrict} if we require that the compositions $H\hat C
\cong \Pi K \bar C \cong \Pi \Sigma H \hat C$ and $K\bar C \cong
\Sigma H \hat C\cong \Sigma\Pi K\bar C$ give the unit $\iota$ and the
counit $\eps$.

\pskip \textbf{The natural transformation $\bm{\delta:H\Pi\to \Pi K}$}
is obtained by extending (\ref{equ:HK-H}) from $\ccal$ to $\Ind\ccal$
as follows.
$\Pi X$ is a filtered colimit $\hat C_i\to \Pi X$.  If $H$ preserves
filtered colimits we therefore obtain $H\Pi\to \Pi K$ as in
\begin{equation}\label{equ:def-delta} 
\xymatrix@C=50pt{ 
 \Pi X &  H\Pi X \ar[r]^{\delta_X} & \Pi K X \\ 
 \hat C_i\ar[u]^{c_i}& H \hat C_i \ar[u]^{H c_i} \ar[r]^{\eqref{equ:HK-H}} & \Pi K \Sigma 
 \hat C_i \ar[u]_{\Pi K c_i^\sharp}  
} 
\end{equation} 
where $c_i^\sharp: \Sigma \hat C_i \to X$ is the transpose of
$c_i:\hat C_i\to \Pi X$. $\delta$ allows us to lift $\Pi$ to a functor
\begin{equation}\label{equ:tildePi} 
\xymatrix{ 
{\Alg(H)} & & {\Alg(K)} \ar@/_/[ll]_{\tilde\Pi}  
} 
\end{equation} 
mapping a $K$-algebra $(B,\beta)$ to the $H$-algebra $(\Pi 
B,\Pi\beta\circ\delta_B)$. 
 
\begin{lem}
  For all $C\in\ccal$ we have
\begin{equation}\label{equ:def-delta-2} 
\xymatrix{ 
H\Pi \Sigma\hat C \ar[rr]^{\delta_{\Sigma\hat C}}& & \Pi K \Sigma\hat C \\ 
& H\hat C \ar[ul]^{\iota_{\hat C}}  \ar[ur]_{\kappa_C}  
} 
\end{equation} 
\end{lem}

\pskip \textbf{The transpose $\bm{\delta^*:\Sigma H\to K\Sigma$}} of
$\delta$ is defined as
\begin{equation}\label{equ:def-delta*} 
  \xymatrix{  \delta^* \ = \ \Sigma H\ar[r] & \Sigma H\Pi\Sigma
    \ar[r]^{\Sigma\delta\Sigma}  & \Sigma\Pi K\Sigma \ar[r] & K\Sigma } 
\end{equation} 
where the unlabelled arrows arise from the unit and counit of the
$\Sigma\dashv\Pi$.  We will show below that $H$-algebras can be
presented as $K$-algebras if there is some, not necessarily natural,
\begin{equation}\label{eq:h}
  h:K\Sigma\to \Sigma H
\end{equation}
such that
\begin{equation}
  \label{eq:delta-h}
  h\circ \delta^*  = \id
\end{equation}

\pskip \textbf{The transformation $\bm{h:K\Sigma\to \Sigma H}$} may
not exist in general, but we can say more if $K$ restricts to
$\ccal$. Then we may require the existence of an $h$ as in the
following diagram
\begin{equation}\label{equ:def-h} 
\xymatrix@C=50pt{ 
  {} A   & K\Sigma A \ar[r]^{h_A} &  \Sigma HA \\ 
  \hat A_k\ar[u]^{d_k}  & K\Sigma \ar[u]_{}\hat A_k \ar[r]^{\eqref{equ:HK-K}}&
  \Sigma H \hat A_k \ar[u]^{} 
} 
\end{equation} 
where the $d_k$ are a filtered colimit. Moreover, the transformation
$h$ does exists if $K$ weakly preserves filtered colimits. We don't
require that $h_A$ be uniquely determined or natural. $h$ allows us to
lift $\Sigma$ to a map on objects
\begin{equation}\label{equ:tildeSigma} 
  \xymatrix{ 
    {\Alg(H)} 
    \ar@/_/[rr]_{\tilde\Sigma}  & &  
    {\Alg(K)} 
  } 
\end{equation}

\begin{lem}\label{lem:delta*}
  If $h$ is as in Diagram~\ref{equ:def-h} and $K$ restricts to
  $\ccal$, see \eqref{equ:HK-restrict}, then \eqref{eq:delta-h} holds.
\end{lem}

\begin{proof}
  We first not that the lower row of \eqref{equ:def-h} is the inverse
  $(\delta^*_{A_k})^{-1}$ of the iso $\delta^*_{A_k}$. This is a
  direct consequence of \eqref{equ:HK-restrict} and
  \eqref{equ:def-delta-2}.  Since $\delta^*_{A_k}$ is natural and
  \eqref{equ:def-h} commutes, it follows $h_A\circ \delta^*_A \circ
  \Sigma H d_k = \Sigma H d_k \circ (\delta^*_{A_k})^{-1}\circ
  \delta^*_{A_k}$, hence $h_A\circ \delta^*_A \circ \Sigma H d_k =
  \Sigma H d_k $. Now \eqref{eq:delta-h} follows from $\Sigma H d_k$
  being a colimit.
\end{proof}

\begin{rem}
  Part III will be devoted to the logical interpretation of the
  developments of this section. But let us say here already that $H$
  will represent the syntax of a modal logic, $K\op$ its coalgebraic
  models, and $\delta$ will map a formula to its denotation, that is,
  to a set of states. An element of $\Sigma H$ will be a maximal
  consistent theory $\Phi$ and $h\op:\Sigma\op H\op\to K\op\Sigma\op$
  will map a theory to a state $x$. Then \eqref{eq:delta-h}, that is,
  $(\delta^*)\op\circ h\op=\id$, ensures that the theory of $x$
  coincides with $\Phi$. In other words \eqref{eq:delta-h} says that
  every maximal consistent one-step theory has a one-step model.
\end{rem}

\medskip\noindent\textbf{Representing $H$-algebras as $\Pi$-images of
  $K$-algebras.  }  Denote by $\iota$ the unit of the adjunction
$\Sigma\dashv\Pi$. Our next theorem states that for all algebras
$\alpha:HA\to A$ the following diagram commutes
\begin{equation}\label{equ:jonsson-tarski} 
\xymatrix@C=20pt{ 
   A \ar[d]_{\iota_A} & & & HA \ar[lll]_{\alpha}\ar[d]^{\ H\iota_A} \\ 
  \Pi\Sigma A & \Pi\Sigma HA \ar[l]_{\Pi\Sigma\alpha\ } &  
    \Pi K\Sigma A \ar[l]_{\ \Pi h_A} & H\Pi\Sigma A 
  \ar[l]_{\delta_{\Sigma A}} } 
\end{equation} 
 
\begin{thm}\label{thm:IndPro} 
  Suppose in Diagram~\ref{equ:IndProAlg} that $H$ preserves filtered
  colimits and that $H$ is determined by $K$ on $\ccal$, ie
  \eqref{equ:HK-H} holds.
\begin{enumerate}[\em(1)]
\item Assume there is $h:K\Sigma\to \Sigma H$ satisfying
  \eqref{eq:delta-h}.  Then for any $H$-algebra $(A,\alpha)$ we have
  that $\iota_A:A\to\Pi\Sigma A$ is an $H$-algebra morphism
  $(A,\alpha)\to\tilde\Pi(\Sigma A,\alpha\circ h_A)$.
\item Furthermore, if $K$ restricts to $\ccal$, see
  \eqref{equ:HK-restrict}, and weakly preserves filtered colimits,
  then there is an $h$ satisfying the assumption of item 1.
\item If, morevoer, $K$ preserves filtered colimits then $h$ is
  uniquely determined by \eqref{equ:def-h} and a natural
  transformation. 
\end{enumerate}
\end{thm}
 
\begin{proof}
  The second item is immediate from Lemma~\ref{lem:delta*}, with the
  existence of $h$ coming from $K$ mapping the colimit $\Sigma d_k$ to
  a weak colimit. For the third item, we note that if $K\Sigma d_k$ is
  even a colimit, then $h$ is uniquely determined, which in turn
  yields naturality. Thus it remains to prove that that
  Diagram~\ref{equ:jonsson-tarski} commutes. Since $\iota$ is natural
  $\iota_{HA}=\Pi h_A\circ \delta_{\Sigma A} \circ H\iota_A$ does
  suffice,
\begin{equation}\label{equ:figure1} 
\xymatrix@C=25pt@R=35pt{ 
  HA \ar[r]_{H\iota_A}\ar@/^20pt/[rrrrr]^{\iota_{HA}} & H\Pi\Sigma
  A\ar[rr]_{\delta_{\Sigma A}} && \Pi K \Sigma A \ar[rr]_{\Pi h_A } &&
  \Pi\Sigma H A \ar@/^20pt/[ll]^{\Pi\delta^*_A}
}
\end{equation} 
for which in turn, because of \eqref{eq:delta-h}, it suffices to have
$\delta_{\Sigma A}\circ H\iota_A = \Pi\delta^*_A\circ \iota_{HA}$. For
this we first note that in the following diagram the rectangles
consisting of non-dotted arrows commute due to 
$\iota:\Id\to\Pi\Sigma$ being natural.
\begin{equation}
\xymatrix@C=50pt@R=35pt{ 
\Pi\Sigma H \ar[r]^{\Pi\Sigma H \iota }\ar@(ur,ul)@{..>}[rrr]^{\Pi\delta^* } & 
\Pi\Sigma H \Pi\Sigma \ar[r]^{ \Pi\Sigma\delta\Sigma} & 
\Pi\Sigma\Pi K \Sigma \ar@{..>}[r]^{\Pi\eps K\Sigma} & \Pi K \Sigma \\
H \ar[r]^{H\iota }\ar[u]^{\iota H} & 
H\Pi\Sigma\ar[r]^{\delta\Sigma}\ar[u]^{\iota H\Pi\Sigma }  & 
\Pi K \Sigma \ar[u]^{\iota \Pi K\Sigma}\ar@{..>}[ur]^{\id}  
}
\end{equation} 
Further, with $\eps$ denoting the counit of $\Sigma\dashv\Pi$, the
triangle commutes due to the definition of adjunction. We have shown
$\delta_{\Sigma A}\circ H\iota_A = \Pi\delta^*_A\circ \iota_{HA}$.
\end{proof}

\noindent The theorem does not imply that $\iota_A:A\to\Pi\Sigma A$ is
a monomorphism. But this holds in case that $\Ind\ccal$ is $\BA$ or
the category $\DL$ of distributive lattices as in the following
example.
  
\begin{example}\hfill
\begin{enumerate}[(1)]
\item We obtain the setting of J{\'o}nsson and
  Tarski~\cite{jons-tars:bao1} with $\ccal=\BA_\omega$ as in
  Example~\ref{exle:IndPro}, $H$ the functor $L$ from
  Example~\ref{exle:mod-alg} and $K$ the powerset. With this data, our
  theorem states that every Boolean algebra with operators can be
  embedded into a complete Boolean algebra whose carrier is a
  powerset.
 
\item We obtain the setting of Gehrke and
  J{\'o}nsson~\cite{gehr-jons:DLO} with $\ccal$ finite distributive
  lattices. For $H$ one can take, for example, the Vietoris functor of
  Johnstone~\cite{johnstone:vietoris-locales}, restricted to $\DL$, and
  for $K$ the convex powerset functor on posets.
\end{enumerate}
\end{example}

\begin{rem}
  Compared to the earlier version of the paper, we reorganised the
  proof of the theorem and made condition \eqref{eq:delta-h}
  explicit. This allows us to strengthen the statement of the theorem
  and also to compare it precisely to \cite[Theorem 3]{kkp:calco05},
  which is now the special case of Theorem~\ref{thm:IndPro}.1 where
  $\ccal=\BA_\omega$. Indeed, for $\ccal=\BA_\omega$, the existence of
  an $h$ satsifying \cite[Definition 1]{kkp:calco05} is equivalent to
  the existence of an $h$ satisfying \eqref{eq:delta-h}. The
  categorical formulation \eqref{eq:delta-h} of this condition using
  the transpose $\delta^*$ (whose importance for coalgebraic logic was
  shown by Klin~\cite{klin:mfps07} where it is called $\rho^*$) is
  new.
\end{rem}

 \section{Introduction to Part III: Functorial Coalgebraic Logic}

\noindent%\textbf{Overview. }
We develop the point of view that if a coalgebra is given wrt a
functor $T:\Set\to\Set$ then a (finitary, classical) modal logic is
given by a functor $L:\BA\to\BA$ on Boolean algebras\footnote{As
  opposed to Part II, this part will benefit from a notation working
  with contravariant functors $P:\Set\to\BA$ and $S:\BA\to\Set$
  instead of covariant functors $\Pi:\Set\op\to\BA$ and
  $\Sigma:\BA\to\Set\op$.}
\begin{equation}
\xymatrix{ 
{\BA} \ar@(dl,ul)[]^{L \ } 
  \ar@/_/[rr]_{S}  & &  
{\Set}\ar@/_/[ll]_{P} \ar@(dr,ur)[]_{\ T}  
} 
\end{equation} 
together with the semantics
$$\delta:LP\too PT.$$
We call such a logic $(L,\delta)$ abstract, because no concrete
syntactic description has been fixed. Such a concrete description
arises from a presentation: A presentation of $\BA$ describes $\BA$ as
a category $\Alg(\Sigma_\BA,E_\BA)$ of algebras for a signature and
equations in the usual sense; this gives us classical propositional
logic with connectives from $\Sigma_\BA$ and axioms from
$E_\BA$. Moreover, if $L$ has a presentation
$\langle\Sigma_L,E_L\rangle$, this gives us the concrete logic with
operation symbols from $\Sigma_\BA+\Sigma_L$ and axioms $E_\BA+E_L$, see Theorem~\ref{thm:present-fun}.

\medskip\noindent Section~\ref{sec:ml-set} will show how to define
$(L_T,\delta_T)$ from an arbitrary functor $T:\Set\to\Set$ and then
give conditions on $T$ under which the logic $(L_T,\delta_T)$ is
strongly complete (for the global consequence relation).
From Part I, we will use that all $(L_T,\delta_T)$ have a
presentation, ie, the abstract logic indeed arises from a concrete
logic. From Part II, we will use the J\'onsson-Tarski theorem which
will provide us,  as a corollary, with the strong completeness result.

\medskip\noindent To keep this part self-contained, the remainder of
this section contains preliminaries on coalgebras and the next section
details carefully the relationship between abstract logics
$(L,\delta)$ and concrete logics given by presentations.

\bigskip\noindent\textbf{Preliminaries on Coalgebras. }
Coalgebras for a functor provide a uniform account of different kinds
of transition systems and Kripke structures.

\begin{definition}[coalgebra]
  The category $\Coalg(T)$ of coalgebras for a functor $T$ on a
  category $\xcal$ has as objects arrows $\xi:X\to TX$ in $\xcal$ and
  morphisms $f:(X,\xi)\to(X',\xi')$ are arrows $f:X\to X'$ such that
  $Tf\circ\xi=\xi'\circ f$.
\end{definition}

The paradigmatic example are coalgebras $X\to\Pow X$ for the powerset
functor. They can be considered as a set $X$ with a relation
$R \subseteq X\times X$, ie as (unlabelled) transition systems or Kripke
frames. Similarly, $X\to\Pow(C\times X)$ is a transition system with
transitions labelled with elements of a constant set $C$. If $2$
denotes some two-element set, then $X\to 2\times X^C$ is a
deterministic automaton with input from $C$ and a labelling of states
as accepting/non-accepting. To cover all these examples and many more
we can consider the following inductively defined class of `type
functors'.

\begin{example}[gKPF]\label{exle:T-Set}
  A generalised Kripke polynomial functor (gKPF) $T:\Set\to\Set$ is
  built according to
$$T::= \Id \mid K_C \mid T+T \mid T\times T \mid T\circ T \mid \Pow \mid \hcal $$
where $\Id$ is the identity functor, $K_C$ is the constant functor
that maps all sets to a finite set $C$, $\Pow$ is covariant powerset
and $\hcal$ is $2^{2^-}$. 
\end{example}

\begin{rem}
  The term `Kripke polynomial functor' was coined in R\"o\ss iger
  \cite{roessiger:cmcs00}. We add the functor $\hcal$.
  $\hcal$-coalgebras are known as neighbourhood frames in modal logic
  and are investigated, from a coalgebraic point of view, in Hansen
  and Kupke~\cite{hans-kupk:cmcs04}.
\end{rem}

If we can consider the carriers $X$ of the coalgebras to have elements, ie if there is a
forgetful functor $\xcal\to\Set$, each functor $T$ induces a
corresponding notion of bisimilarity or behavioural equivalence.

\begin{definition}[bisimilarity]\label{def:bisim}
  Two states $x_i$ in two coalgebras $X_i$ are $T$-\emph{bisimilar} if
  there is a coalgebra $(X',\xi')$ and there are coalgebra morphisms
  $f_i:(X_i,\xi_i)\to (X',\xi')$ such that $f_1(x_1)=f_2(x_2).$
\end{definition}

\begin{rem}
  In other words, two states are bisimilar if they are in the same
  equivalence class of the equivalence relation generated by pairs
  $(x,f(x))$ where $f$ ranges over all coalgebra morphisms. More
  categorically, two states are bisimilar if they are in the same
  connected component of the category of elements of
  $U:\Coalg(T)\to\Set$. If $U$ has a colimit $Z$, then $Z$ classifies
  $T$-bisimilarity and is the carrier of the final coalgebra.

  This notion of bisimilarity has sometimes been called
  \emph{behavioural equivalence}, since only for weak pullback
  preserving functors it is the case that behavioural equivalence is
  characterised by coalgebraic bisimulations \cite{rutten:uc-j}. On
  the other hand, in cases where the functor $T$ does not preserve
  weak pullbacks, coalgebraic bisimulations are not well-behaved and
  it has been argued since \cite{kurz:phd}, but see also
  \cite{hans-kupk-pacu:calco07-j} for a study of $2^2$-coalgebras,
  that behavioural equivalence is the better notion in such
  situations. We thus find it defensible to choose the more
  recognisable name of bisimilarity for  behavioural
  equivalence.
\end{rem}

In all of the examples above, coalgebraic bisimilarity coincides with
the `natural' notion of equivalence. For, $T=\Pow$ this goes back to
Aczel~\cite{aczel:nwfs,acze-mend:fct}, for deterministic automata two
states are bisimilar iff they accept the same language
(Rutten~\cite{rutten:uc-j}) and $\hcal$-coalgebras have been
investigated by Hansen and Kupke~\cite{hans-kupk:cmcs04}.
We now turn to logics for $T$-coalgebras.

\section{Functorial Modal Logics for Coalgebras}
\label{sec:ml-abs-concr}
\noindent In this section we present a general framework for logics
for $T$-coalgebras. We do this in two steps.

\begin{enumerate}[(1)]
\item First, abstracting from syntax, we simply consider as formulas
  of the logic the elements of the initial $L$-algebra, where $L$ is a
  functor which is dual to $T$ in a suitable sense.

\item Second, we obtain a syntax and a proof system for the abstract
  logic from a presentation of the functor $L$. We call these logics
  the concrete logics of $T$-coalgebras.
\end{enumerate}

\noindent The point of this separation is that it allows us to prove
results about concrete logics in a presentation-independent way on the
level of the abstract logics. An example of this is presented in the
next section.

\subsection{Abstract modal logics}
 
\noindent We are interested in the following situation
\begin{equation}\label{diag:XA} 
\xymatrix{ 
  {\Coalg(T)} \ar@/^/[r]^{\tilde P} \ar[d]_{}  
  & {\Alg(L)\ } %\ar@/^/[l]^{\tilde S}  
    {\ar[d]^{}}\\ 
  {\ \xcal} \ar@(dl,ul)[]^{T} \ar@/^/[r]^{P} \ar[d]_{} 
  & {\acal\ } \ar@/^/[l]^{S} \ar@(dr,ur)[]_{L} {\ar[d]^{U}}\\ 
  {\Set} & {\Set} %\ar@/^/[u]^{F} 
} 
\end{equation} 
where $P$ and $S$ are contravariant functors.

\begin{example}\label{exle:XA}\hfill
\begin{enumerate}[(1)]
\item $\xcal=\Set$, $\acal=\BA$. $PX$ is the powerset of $X$ and $SA$
  is the set of ultrafilters on $A$. On maps both $P$ and $S$ act as
  inverse image. It is also useful to think of $PX$ as the set
  $\Set(X,2)$ of functions from $X$ to a two-element set $2$ and to
  think of $SA$ as the set $\BA(A,\twobb)$ of algebra morphisms from
  $A$ to the two-element Boolean algebra $\twobb$.  
\item $\xcal$ is the category $\Stone$ of Stone spaces (compact
  Hausdorff spaces that have a basis of clopens), $\acal=\BA$, $PX$ is
  the set of clopens of $X$ and $SA$ is the space of ultrafilters on
  $A$ with a basis given by $\{\{u\in SA\mid a\in u\} \mid a\in
  A\}$. In this situation, $\xcal$ and $\acal$ are dually equivalent.
\end{enumerate}
\end{example}

\noindent Kripke frames arise under (1) and descriptive (general) frames \cite{goldblatt:mm,brv:ml}
under (2). The latter situation has been studied from a coalgebraic point of view in
\cite{kkv:cmcs03-j}, whereas this paper will focus on the former.

\begin{rem}\label{rmk:XA}
  Diagram \ref{diag:XA} has too many possible variations to
  give---at this stage---an axiomatic account of the properties the
  data in (\ref{diag:XA}) should satisfy in order to give rise to
  coalgebraic logics. We indicate some of the possible variations.
\begin{enumerate}[(1)]
\item In Example~\ref{exle:XA}.(1) above, one could keep $\xcal=\Set$
  but take $\acal$ to be eg distributive lattices or semi-lattices. A
  sufficient set of conditions for this set-up is the following:
  $\xcal=\Set$ and $\acal$ any variety such that there is
  $P:\Set\to\acal$ with $UPX=2^X$. It then follows that $P$ has an
  adjoint $SA=\acal(A,P1)$ but one would want to require that the unit
  $\iota_A:A\to PSA$ is an embedding.
\item In Example~\ref{exle:XA}.(2) above, one could work with other
  dualities such as the one of spectral spaces and distributive
  lattices.
\item One could also replace $\Set$ by some other categories such as
  $\mathsf{Poset}$.
\end{enumerate}
\end{rem}

\noindent To continue the discussion of the data in (\ref{diag:XA}),
we assume that $\acal$ is a variety in the sense of Part I and that
$L$ is a sifted colimit preserving functor on $\acal$, that is, $L$ is
determined by its action on finitely generated free algebras. Then the
forgetful functor $U^a:\Alg(L)\to\Set$ has a left adjoint $F^a$ and we
consider $U^aF^aV$ as the set of formulas of $L$ in propositional
variables $V$. The semantics of $L$ in terms of $T$-coalgebras is
specified by choosing a natural transformation
\begin{equation}\label{equ:delta}
\delta:LP\to PT,
\end{equation}
where we assume, as in Remark~\ref{rmk:XA}.(1), that $P$ is a functor
satisfying $UPX=2^X$. Intuitively, $\delta$ takes syntax from $LPX$
and maps it to its interpretation as a subset of $TX$. Technically,
$\delta$ allows us to extend the functor $P:\xcal\to\acal$ to a
functor $\tilde P:\Coalg(T)\to\Alg(L)$, where $\tilde P$ maps a
coalgebra $(X,\xi)$ to the $L$-algebra
\begin{equation}\label{equ:tildep}
P\xi\circ\delta_X\ : \ LPX\to PTX \to PX
\end{equation}
Consequently, every formula, ie every element of the free $L$-algebra
$F^aV$, has a unique interpretation as an element of $PX$, ie a subset
of $X$. This is summarised in the following definition.
\begin{definition}\label{def:abstract-logic} 
  Let $T:\Set\to\Set$ be a functor, let $L:\acal\to\acal$ be a sifted
  colimit preserving functor on a variety $U:\acal\to\Set$, and let
  $P:\Set\to\acal$ be a contravariant functor satisfying
  $UPX=2^X$. Further, let $F^a$ be a left-adjoint of the forgetful
  functor $U^a:\Alg(L)\to\Set$ and let $\delta:LP\to PT$ be a natural
  transformation. We call $(L,\delta)$ an (abstract) logic for
  $T$-coalgebras. The formulas of the logic are the elements of
  $U^aF^aV$. Given a coalgebra $(X,\xi)$, we write
  $\sem{-}_{(X,\xi,h)}$ for the morphism $F^aV\to\tilde P(X,\xi)$
  determined by the valuation $h:V\to UPX$.  We define
$$(X,\xi,h)\models\phi\lesssim\psi$$ 
if $\sem{\phi}_{(X,\xi,h)}\subseteq \sem{\psi}_{(X,\xi,h)}$. For a
collection $\Gamma$ of `sequents' $\{\phi_i\lesssim\psi_i\mid i\in
I\}$, we write
\begin{equation}\label{eq:global-consequence}
\Gamma\models(\phi\lesssim\psi)
\end{equation}
for the global consequence relation, that is, if for all
$T$-coalgebras $(X,\xi)$ and all valuations $h$ we have that
$(X,\xi,h)\models\Gamma$ only if $(X,\xi,h)\models
(\phi\lesssim\psi)$. We also write $\models (\phi\lesssim\psi)$ for
$\emptyset\models(\phi\lesssim\psi)$.
\end{definition}

\begin{rem}
  For $\acal$ being semi-lattices or distributive lattices
  $\phi\lesssim\psi$ can be rendered as the equation
  $\phi\wedge\psi=\phi$. In case $\acal=\BA$, since Boolean algebra
  has implication, it is enough to consider sequents of the form
  $\top\lesssim\psi$. In this case we drop the `$\top\lesssim$' and write
  $\models\psi$, etc.
\end{rem}
 
\begin{prop}\label{prop:abstract-logic} 
  The logic for $T$-coalgebras given in
  Definition~\ref{def:abstract-logic} respects bisimilarity.
\end{prop}

\begin{proof}
  Let $f:(X,\xi)\to(X',\xi')$ be a coalgebra homomorphism and let
  $h:V\to UPX$, $h':V\to UPX'$ be two valuations such that $UPf\circ
  h'=h$. According to Definition~\ref{def:bisim}, we have to show that
  $x\in\sem{\phi}_{(X,\xi,h')} \ \Leftrightarrow \
  f(x)\in\sem{\phi}_{(X',\xi',h')}$. But this is immediate from the
  universal property of $F^aV$.
\end{proof}
 
\newcommand{\modop}{\heartsuit}
\newcommand{\arity}{\mathrm{arity}}
\newcommand{\sconcrete}{^{\mathit{concrete}}}
\newcommand{\sabstract}{^{\mathit{abstract}}}

\subsection{Concrete modal logics}\label{sec:concr-log}
 
\noindent We restrict our attention now to set-coalgebras and to
logics over $\BA$. Fix a set of operations $\Sigma_\BA$ and equations
$E_\BA$ describing $\BA$, that is, $\BA=\Alg(\Sigma_\BA, E_\BA)$. We
assume that the constants $\bot,\top$ are in $\Sigma_\BA$.

Conceptually, a concrete logic is given by a finitary presentation
$\langle \Sigma,E \rangle$ of a functor $L$ in the sense of
Definition~\ref{def:present-fun} together with a natural
transformation $LP\to PT$ as in (\ref{equ:delta}). Explicitely, this
means that a concrete logic is given by the following data.

\begin{definition}[concrete logic for
  $T$-coalgebras]\label{def:clogic} 
  A concrete logic for $T$-coalgebras is given by a triple
  $(\Sigma,E,\Delta)$ as follows.
\begin{desCription}
\item\noindent{\hskip-12 pt\bf modal operators:}\ A set $\Sigma$ of operation symbols and a map
  $\arity:\Sigma\to\omega$ assigning to each operation symbol a finite
  arity.
\item\noindent{\hskip-12 pt\bf equations (axioms):}\ A set $E$ of equations $s=t$ of rank 1 as in
  Definition~\ref{def:rank1}. That is, $s,t$ are terms over
  $\Sigma_\BA+\Sigma$ and variables $V$ in which each variable is in
  the scope of precisely one modal operator (operation symbol from
  $\Sigma$).
\item\noindent{\hskip-12 pt\bf semantics:}\ A set $\Delta$ containing for each $\modop\in\Sigma$
  a natural transformation, also called a `predicate lifting',
  \begin{equation}\label{equ:sem-def-clogic}
\sem{\modop}:(2^{\arity(\modop)})^X\to
    2^{TX}.
  \end{equation}
\end{desCription}
The equations $E$ are required to be sound with respect to the
semantics $\Delta$ in the following sense. We lift $\sem{-}$ from
modal operators $\modop$ to terms $s$ of rank 1,
$$\sem{s}: (2^{|V|})^X\to 2^{TX}. $$
In detail, given a valuation $h:V\to 2^X$ and a term $s$, define
$\sem{s}_h\subseteq TX$ inductively as follows. First, $h$ lifts to a
function $\bar h$ on $\BA$-terms by interpreting Boolean operations
set-theoretically. Modal operators are then interpreted according to
\begin{equation}\label{equ:semh}
  \sem{\modop(s_1,\ldots s_{\arity(\modop)})}_h =
  \sem{\modop}(\bar h(s_1), \ldots \bar h(s_{\arity(\modop)})) 
\end{equation}
Then an equation $s=t$ in variables from $V$ is sound, if
$\sem{s}_h=\sem{t}_h$ for all $h:V\to 2^X$.
\end{definition}

More examples will be given in the next section, here we only present
the fundamental one \cite[Def's 1.9, 1.13]{brv:ml}, which
translates into our setting as follows.

\begin{example}\label{exle:bml}
  The basic modal logic for $\pcal$-coalgebras (Kripke
  frames) is given by
\begin{desCription}
\item\noindent{\hskip-12 pt\bf modal operators:}\ one unary operator $\Box$,
\item\noindent{\hskip-12 pt\bf equations (axioms):}\ two equations:
$\Box\top=\top$ with $V=\emptyset$ and $\Box(a\wedge b)= (\Box a)
\wedge (\Box b)$ with $V=\{a,b\}$,
\item\noindent{\hskip-12 pt\bf semantics:}\ 
\begin{align}
  \sem{\Box}  : 2^X &\to 2^{\pcal X} \label{equ:bml}\\
                Y   &\mapsto \{Z\subseteq X \mid Z\subseteq Y\}.
\end{align}
\end{desCription}
\end{example}

\begin{rem}\label{rmk:clogic}\hfill
\begin{enumerate}[(1)]
\item The semantics (\ref{equ:sem-def-clogic}) can be written, in the
  notation of Diagram~\ref{diag:XA}, as a natural transformation
  $U(PX)^n\to UPTX$. This gives a notion of predicate lifting for
  other categories than $\Set$ such as $\Stone$ and what follows
  applies to this setting as well.
\item The signature $\Sigma$ is just a collection of $n$-ary modal
  operators in the usual sense of modal logic, called a similarity
  type in \cite[Def 1.11]{brv:ml}.
\item In modal logic, axioms are usually given by formulas, not by
  equations. The translation between the two formats is a standard
  procedure \cite[Section 5.1]{brv:ml}. In a nutshell, each term in
  operation symbols from $\Sigma_\BA+\Sigma$ is considered as a
  formula. Equations $s=t$ are turned into formulas $s\biimp
  t$. Conversely, any formula $s$ can be read as an equation $s=\top$.
\item Without restricting $E$ to rank 1 the interpretation $\sem{s}_h$
  would not be well-defined.
\item \label{rmk:clogic:i4} The coalgebraic semantics (see below) of
  modal operators in terms of predicate liftings goes back to
  Pattinson~\cite{pattinson:cml-j} and, in the $n$-ary case, to
 Schr\"oder~\cite{schroeder:fossacs05}.
\end{enumerate}
\end{rem}

\noindent The definition of the language below is standard, see
\cite[Def 1.12]{brv:ml}. For the proof system we use equational logic,
see \cite[Def B.20]{brv:ml}.

\begin{definition}[language, proof system]\label{def:lang-proofsyst}
  Let $(\Sigma,E,\Delta)$ be a logic for $T$-coalgebras. The language
  $\lcal(\Sigma,E)$ is the set of terms built from operations
  $\Sigma_\BA+\Sigma$ and the variables that appear in $E$. Terms are
  also called formulas. The proof system is that of equational logic
  plus the additional equations $E_\BA+E$ and we write
  $\vdash_{(\Sigma,E)} s=t$ if an equation is derivable. We also write
  $\vdash_{(\Sigma,E)} \phi$ if $\phi$ is a formula and
  $\vdash_{(\Sigma,E)} \phi=\top$.
\end{definition}

\begin{rem}
  In modal logic, the standard proof system is not equational logic,
  but the two systems are equivalent in terms of the theorems that can be
  derived, see Chapter~5 and Appendix~B of \cite{brv:ml} for full
  details.
\end{rem}

The next definition reformulates \cite[Def 5.19]{brv:ml} in our
notation. As in Definition~\ref{def:lang-proofsyst}, the semantic
component $\Delta$ is not needed here, but only in
Definition~\ref{def:coalg-sem}.

\begin{definition}[algebraic semantics]\label{def:alg-sem}
  Let $(\Sigma,E,\Delta)$ be a logic for $T$-coalgebras. The category
  of modal algebras of $(\Sigma,E,\Delta)$ is the category
  $\Alg(\Sigma_\BA+\Sigma,E_\BA+E)$ of algebras given by the signature
  $\Sigma_\BA+\Sigma$ and satisfying the equations $E_\BA+E$.

  In particular, there is a map $[-]:\lcal(\Sigma,E)\to UFV$ taking
  formulas to the carrier $U^cF^cV$ of the free
  $\Alg(\Sigma_\BA+\Sigma,E_\BA+E)$-algebra $F^cV$ over the variables
  $V$.
\end{definition}

Next we give the coalgebraic semantics of a concrete logic.

\begin{definition}[coalgebraic semantics]\label{def:coalg-sem}
  Let $(\Sigma,E,\Delta)$ be a logic for $T$-coalgebras and let
  $(X,\xi)$ be a $T$-coalgebra. Then
  $\sem{\phi}\sconcrete_{(X,\xi,h)}$ is defined by induction over
  $\phi\in\lcal(\Sigma,E)$ with Boolean clauses as usual and
\begin{equation}\label{equ:semmodop-coalg}
  \sem{\modop(\phi_1,\ldots\phi_{\arity(\modop)})}\sconcrete_{(X,\xi,h)} =
  P\xi\circ\sem{\modop}(\sem{\phi_1}\sconcrete_{(X,\xi,h)}, \ldots 
  \sem{\phi_{\arity(\modop)}}\sconcrete_{(X,\xi,h)}) 
\end{equation}
for each $\modop\in\Sigma$.
\end{definition}

\begin{rem}
  If  the semantics of an $n$-ary modal operator $\modop$ is expressed with the help of the Yoneda lemma
  by a map ${T(2^n)}\to 2$, then
  (\ref{equ:semmodop-coalg}) takes a list of $n$-ary predicates
  $\phi:X\to 2^n$ and maps it
  to $$X\stackrel{\xi}{\too}TX\stackrel{T\phi}{\too}T(2^n)\too 2.$$
\end{rem}

\begin{example}
  Going back to Example~\ref{exle:bml}, we take now $T=\pcal$ so that
  $(X,\xi)$ is a Kripke frame and $\xi(x)$ is the set of successors of
  $x$. Recalling that $P\xi=\xi^{-1}$ it follows immediately from the
  definitions that instantiating (\ref{equ:semmodop-coalg}) with
  (\ref{equ:bml}) gives
$$\sem{\Box\phi}\sconcrete_{(X,\xi,h)}=\{x\in X \mid \xi(x) \subseteq \sem{\phi}\sconcrete_{(X,\xi,h)}\},$$ 
which is the usual definition of the semantics of $\Box$.
\end{example}

In Definition~\ref{def:coalg-sem}, the $\sem{\modop}\in\Delta$
provided the semantics of the modal operators. Alternatively, we can
think of $\Delta$ as giving us a functor from coalgebras to algebras,
mapping a coalgebra to its `complex algebra' \cite[Def
5.21]{brv:ml}. That these two points of view are essentially the same
is the contents of Proposition~\ref{prop:sconcrete-alg-coalg}.

\begin{definition}[complex algebra]\label{def:complex-alg}
  Let $(\Sigma,E,\Delta)$ be a logic for $T$-coalgebras and let
  $(X,\xi)$ be a $T$-coalgebra. Then the complex algebra $\tilde
  P^c(X,\xi)$ of $(X,\xi)$ is the
  $\Alg(\Sigma_\BA+\Sigma,E_\BA+E)$-algebra with carrier $PX$ and
  which interprets operations $\modop\in\Sigma$ according to
\begin{equation}\label{equ:complex-alg}
  \modop^{\tilde P^c(X,\xi)}(a_1,\ldots a_{\arity(\modop)}) = P\xi\circ\sem{\modop}(a_1,\ldots a_{\arity(\modop)}).
\end{equation}
\end{definition}

The relationship between algebraic and coalgebraic semantics follows
the classical pattern \cite[Prop 5.24, Thm 5.25]{brv:ml}, again
replacing Kripke frames by coalgebras.

\begin{prop}[relationship of algebraic and coalgebraic
  semantics]\label{prop:sconcrete-alg-coalg}
  Let $(\Sigma,E,\Delta)$ be a logic for $T$-coalgebras and let
  $(X,\xi)$ be a $T$-coalgebra. Any valuation $h:V\to 2^X$ induces a
  morphism $\mathit{mng}_h:F^cV\to\tilde P^c(X,\xi)$ from the free
  $\Alg(\Sigma_\BA+\Sigma,E_\BA+E)$-algebra over $V$ to the complex
  algebra of $(X,\xi)$. For all $\phi\in\lcal(\Sigma,E)$ we have
\begin{equation}\label{equ:mng}
  \sem{\phi}\sconcrete_{(X,\xi,h)}=\mathit{mng}_h([\phi])
\end{equation}
Consequently, the equation $\phi=\top$ holds in the algebra $\tilde
P(X,\xi)$ iff $\sem{\phi}\sconcrete_{(X,\xi,h)}=X$.
\end{prop}

\begin{proof}
  The proof is a routine induction as in \cite[Prop 5.24]{brv:ml},
  using (\ref{equ:semmodop-coalg}) and (\ref{equ:complex-alg}).
\end{proof}

\begin{thm}[equivalence of abstract and concrete logics]
\label{thm:equiv-abs-concr-log}
For each abstract logic $(L,\delta)$ there is a concrete logic
$(\Sigma,E,\Delta)$, and for each concrete logic there is an abstract
logic $(L,\delta)$, such that concrete and abstract semantics agree. 
\end{thm}

\begin{proof}
  For the purposes of the proof, write $\sem{-}\sabstract_{(X,\xi,h)}$
  for $\sem{-}_{(X,\xi,h)}$ in Definition~\ref{def:abstract-logic}. As
  before, we denote by $F^cV$ the term algebra over
  $(\Sigma_\BA+\Sigma+V)$-terms quotiented by equations $E_\BA+E$. The
  two statements of the theorem say:
  \begin{enumerate}[(1)]
  \item For each abstract logic $(L,\delta)$ there is a concrete logic
    $(\Sigma,E,\Delta)$ such that, for any left-adjoint $F^a$ of
    $U^a:\Alg(L)\to\Set$, there is concrete isomorphism
    $g:\Alg(L)\to\Alg(\Sigma_\BA+\Sigma,E_\BA+E)$, inducing an
    isomorphism $f:F^c V\to g(F^a V)$, so that for all coalgebras
    $(X,\xi)$ and all formulas $\phi\in\lcal(\Sigma,E)$ we have
    $\sem{f[\phi]}\sabstract_{(X,\xi,h)}=\sem{\phi}\sconcrete_{(X,\xi,h)}$.
  \item For each concrete logic $(\Sigma,E,\Delta)$ there is an
    abstract logic $(L,\delta)$ and a left-adjoint $F^a$ of
    $U^a:\Alg(L)\to\Set$ such that such that for all coalgebras
    $(X,\xi)$ and all formulas $\phi\in\lcal(\Sigma,E)$ we have
    $\sem{[\phi]}\sabstract_{(X,\xi,h)}=\sem{\phi}\sconcrete_{(X,\xi,h)}$.
\end{enumerate}

To prove (1), take the presentation $\langle\Sigma,E\rangle$ of $L$
from Remark~\ref{rmk:pres-from-functor} and let $\Delta$ be given by
$\sem{\modop}:(2^{\arity(\modop)})^X\to 2^{TX}$, $(Y_1, \ldots
Y_{\arity(\modop)})\mapsto \delta_X(\modop(Y_1, \ldots
Y_{\arity(\modop)})$ . 
  %Soundness of the equations $E$ is immediate
  % from the respective definitions. 
  By Theorem~\ref{thm:present-fun}, we have a concrete isomorphism
  $g:\Alg(L)\to\Alg(\Sigma_\BA+\Sigma,E_\BA+E)$, inducing an
  isomorphism $f:F^c V\to g(F^a V)$. Since $g(\tilde
  P(X,\xi)) = \tilde P^c(X,\xi)$ where the two versions of $\tilde
  P$ refer to (\ref{equ:tildep}) and Definition~\ref{def:complex-alg}
  respectively, we have $\sem{f[\phi]}\sabstract_{(X,\xi,h)} =
  \mathit{mng}_h([\phi]) = \sem{\phi}\sconcrete_{(X,\xi,h)}$, where
  the second step is (\ref{equ:mng}).

  For (2), define $L$ as in Remark~\ref{rmk:pres-from-functor} and let
  $\delta_X(\modop(Y_1, \ldots Y_{\arity(\modop)})=\sem{\modop}(Y_1,
  \ldots Y_{\arity(\modop)})$. Since $L$ is not an `absolutely' free
  algebra but quotiented wrt $E$, we need to check that $\delta$ is
  well-defined, but this follows from the equations $E$ being sound,
  see Definition~\ref{def:clogic}. By Theorem~\ref{thm:present-fun},
  we have a concrete isomorphism
  $g:\Alg(L)\to\Alg(\Sigma_\BA+\Sigma,E_\BA+E)$. Choose $F^a$ so that
  $g(F^cV)=F^aV$ and finish the argument as above in item (1).
\end{proof}

\begin{rem}
  To summarise, given a functor $L:\BA\to\BA$ determined by its action
  on finitely generated free Boolean algebras, we can find a
  presentation $\langle\Sigma, E\rangle$ as described in
  Remark~\ref{rmk:pres-from-functor}. This gives us an isomorphism
  between $L$-algebras and $(\Sigma_\BA+\Sigma,E_\BA+E)$-algebras as
  described in Remark~\ref{thm:present-fun:rmk}.
  Conversely, given operations $\Sigma$ and equations $E$ of rank 1,
  we define a functor $L$ as described in
  Remark~\ref{rmk:pres-from-functor} and this gives us, again, an
  isomorphism between $L$-algebras and
  $(\Sigma_\BA+\Sigma,E_\BA+E)$-algebras as described in
  Remark~\ref{thm:present-fun:rmk}. The theorem shows that the logic
  arising from $L$ and the logic arising from $(\Sigma,E)$ are
  equivalent.
\end{rem}

\begin{example}\label{exle:equiv-abs-concr-log}
  Starting from the concrete logic of Example~\ref{exle:bml}, we
  define $(L,\delta)$ as follows. $LA$ is the $\BA$ generated by $\Box
  a, a\in A$ and quotiented with respect to the equations of
  Example~\ref{exle:bml}. To give Boolean algebra homomorphisms
  $\delta_X:LPX\to P\pcal X$ it is enough to describe them on generators,
  which is exactly what $\sem{\Box}$ in \ref{exle:bml} does.

  Conversely, we could start by defining $LA=P\Pow SA$ on finite
  Boolean algebras. This determines $L$ on finitely generated free
  algebras and hence defines a sifted colimit preserving
  functor. Therefore we can present $L$ as in
  Remark~\ref{rmk:pres-from-functor}. This canonical presentation,
  which is made from all (finitary) predicate liftings for $\pcal$, is
  different from the presentation with a single $\Box$, but it
  presents an isomorphic functor. This observation is at the the heart
  of the next section.
\end{example}

\medskip\noindent\textbf{Summary. } The correspondence between
functors and logics gives us the licence to switch at will between the
abstract point of view for which a logic is a pair $(L,\delta)$ and
the concrete point of view for which a logic is given by operations
and equations. We will therefore, in the following, blur the
distinction whenever convenient.

The good functors $L:\BA\to\BA$ for which this correspondence is
available, are those functors for which one of the following
equivalent conditions holds:
\begin{iteMize}{$\bullet$}
\item $L$ preserves sifted colimits,
\item $L$ is determined by its action on finitely generated free
  algebras.
\end{iteMize}
This follows from the general considerations of Part I, but in the
case of $\BA$ one can go further. Using
Proposition~\ref{prop:sifted=filtered-BA}, we can extend this list by
saying that, up to modification of $L$ on the one-element Boolean
algebra, one of the following equivalent condition holds:
\begin{iteMize}{$\bullet$}
\item $L$ preserves filtered colimits,
\item $L$ preserves directed colimits,
\item $LA$ is determined by its action on the finite subalgebras of
  $A$.
\end{iteMize} 

\section{The Finitary Modal Logic of Set-Coalgebras}
\label{sec:ml-set} 
\noindent 
The aim of this section is to associate a modal logic to an arbitrary
functor $T:\Set\to\Set$. As we are interested here in classical
propositional logic the logic will be given by a functor
$L_T:\BA\to\BA$. That is, we are concerned with the following situation
\begin{equation}\label{equ:BA-Set} 
  \xymatrix{ 
    {\BA}  \ar@(dl,ul)[]^{L_T \ } 
    \ar@/_/[rr]_{S}  & &  
    {\Set}\ar@/_/[ll]_{P} \ar@(dr,ur)[]_{\ T}  
  } 
\end{equation} 
where $S$ maps an algebra to the set of its ultrafilters and $P$ is
the contravariant powerset. 
For the readers of Part II, we note that \eqref{equ:BA-Set} is the
instance of \eqref{equ:IndProAlg} with $\ccal$ being the category of
finite Boolean algebras. But in this part, instead of writing arrows
in $\Pro(\BA_\omega)\simeq\Pro(\Set_\omega\op)\simeq\Set\op$, we write
them in $\Set$. Further details of how to translate the notation from
Part II are summarised in the next remark.
\begin{rem}\label{rmk:partii-partiii}
  To apply the results of Part II, instantiate $H=L_T$,
  $\Ind\ccal=\BA$, $\Pro\ccal=\Set\op$, $K=T\op$. We also write
  $P:\BA\to\Set$ for the contravariant functor given by the covariant
  $\Pi:\BA\to\Set\op$ and similarly we write $S:\Set\to\BA$ for the
  contravariant functor given by the covariant
  $\Sigma:\Set\op\to\BA$. Accordingly, the types of the unit and
  counit become $\iota:\Id\to PS$ and $\eps:\Id\to SP$. Similarly, we
  have $\delta:LP\to PT$, $\delta^*:TS\to SL$, $h:SL\to TS$.
\end{rem}
The next definition generalises
Example~\ref{exle:equiv-abs-concr-log}.
\begin{definition}[$(L_T,\delta_T)$] \label{def:LT} Let
  $T:\Set\to\Set$ be any set-functor. We define $L_T$ to be $L_T A =
  PTSA$ on finite BAs. This determines $L_T$ on finitely generated
  free algebras and hence defines a sifted colimit preserving
  functor. For finite $X$ we put $\delta_TX:L_TPX=PTSPX\cong PTX$ and
  extend to arbitrary $X$ as in (\ref{equ:def-delta}).
\end{definition}

\begin{rem}[Bisimulation-somewhere-else]
  Since $\BA$ and $\Stone$ are dually equivalent, the functor
  $L_T:\BA\to\BA$ has a dual $\hat T:\Stone\to\Stone$, which
  simplifies the definition of the same functor in \cite[Def 7, Rmk
  16]{kkp:calco05}. We can associate to any $T$-coalgebra $(X,\xi)$ a
  $\hat T$-coalgebra $SPX\to SPTX\to SL_TPX\to\hat TSPX$. Then two
  states $x_1,x_2\in X$ satisfy the same formulas of $(L_T,\delta_T)$
  iff they are bisimilar (not necessarily in $X$ but) in $SPX$, see
  \cite[Thm 18]{kkp:calco05}.
\end{rem}

This definition applies in particular to all gKPFs, see
Example~\ref{exle:T-Set}, and we are able now to supplement further
examples to Section~\ref{sec:concr-log}. We note that Part II of this
paper does not deal with many-sorted signatures which are required for binary
functors $\BA\times\BA\to\BA$. This has been done in Schr\"oder and
Pattinson~\cite{schr-patt:modular} and for the functorial framework of
this paper in \cite{kurz-petr:cmcs08-j}.
\renewcommand{\Pi}{P}
\renewcommand{\Sigma}{S}

\begin{example}\label{exle:L-BA}
  We describe functors $L:\BA\to\BA$ or $L:\BA\times\BA\to\BA$
  by generators and relations as follows.
\begin{enumerate}[(1)]
\item $L_{K_C}(A)$ is the free $\BA$ given by generators $c\in C$
  and satisfying $c_1\wedge c_2=\bot$ for all $c_1\not=c_2$ and $\bigvee_{c\in
    C} c = \top$.
\item $L_+(A_1,A_2)$ is generated by $[\kappa_1]a_1$,
  $[\kappa_2]a_2$, $a_i\in A_i$ where the $[\kappa_i]$ preserve finite
  joins and binary meets and satisfy $[\kappa_1] a_1\wedge
  [\kappa_2]a_2=\bot$, $[\kappa_1]\top\vee [\kappa_2]\top=\top$,
  $\neg[\kappa_1]a_1=[\kappa_2]\top\vee[\kappa_1]\neg a_1$,
  $\neg[\kappa_2]a_2=[\kappa_1]\top\vee[\kappa_2]\neg a_2$.
\item $L_\times(A_1,A_2)$ is generated by $[\pi_1]a_1$,
  $[\pi_2]a_2$, $a_i\in A_i$ where $[\pi_i]$ preserve Boolean
  operations.
\item $L_\Pow(A)$ is generated by $\Box a$, $a \in A$, and $\Box$
  preserves finite meets.
\item $L_\hcal(A)$ is generated by $\Box a$, $a \in A$ (no equations).
\end{enumerate}
For the semantics, we define Boolean algebra morphisms $\delta_T$
\begin{enumerate}[(1)]
\item $L_{K_C}\Pi X \to \Pi C$ \ by \ $c\mapsto \{c\}$, 
\item $L_+(\Pi X_1,\Pi X_2) \to \Pi(X_1+X_2)$ \ by \ $[\kappa_i]a_i\mapsto a_i$, 
\item $L_\times(\Pi X,\Pi Y) \to \Pi(X_1\times X_2)$ \ by \ $[\pi_1]a_1\mapsto a_1\times X_2$, $[\pi_2]a_2\mapsto X_1\times a_2$, 
\item $L_\Pow \Pi X\to \Pi\Pow X$ \ by \ $\Box a \mapsto \{b\subseteq
  X\mid b\subseteq a\}$,
\item $L_\hcal \Pi X\to \Pi\hcal X$ \ by \ $\Box a \mapsto \{s\in \hcal X
  \mid a \in s \}$.
\end{enumerate}
and extend them inductively to $\delta_T:L_T\Pi\to\Pi T $ for all gKPF
$T$. To be precise, we will for the moment denote by
$(L'_T,\delta'_T)$ the $(L_T,\delta_T)$ given by the presentations in
this example and reserve the notation $(L_T,\delta_T)$ for the logics
given by Definition~\ref{def:LT}. We need to show that
$(L'_T,\delta'_T)$ is equivalent in the sense of
Theorem~\ref{thm:equiv-abs-concr-log} to $(L_T,\delta_T)$, in
other words, that the presentations of this example indeed present the
logics of Definition~\ref{def:LT}. This amounts to showing that
$(\delta'_T)_X: L'_T\Pi X\to \Pi TX$ is an isomorphism for all finite
sets $X$. It is exactly here where the machinery presented in this
paper needs to be supplemented by additional work depending on the
concrete presentation at hand. In our case this is essentially known:
(1)-(3) are slight variations of cases appearing in
Abramsky~\cite{abramsky:dtlf}, (4) is in
Abramsky~\cite{abramsky:cooks-tour}, and $\delta_X$ in (5) is given by
the identity on $2^{2^{2^X}}$.

For gKPFs excluding $\hcal$, the maps 
\begin{equation}\label{eq:h2}
h_A:SLA\to TSA
\end{equation} 
from (\ref{eq:h}) have been described by Jacobs~\cite[Definition
5.1]{jacobs:many-sorted}. We detail the definitions of the following
two cases.
\begin{enumerate}[(1)]\setcounter{enumi}{3}
\item $h_A:\Sigma L_\Pow A\to \Pow\Sigma A$ maps $v\in\Sigma L_\Pow A$
  to $\{u\in\Sigma A \mid \Box a \in v\Rightarrow a \in u\}$.
\item $h_A:\Sigma L_\hcal A\to \hcal\Sigma A$ maps $v\in\Sigma L_\hcal
  A$ to $\{\hat a\in 2^{\Sigma A} \mid \Box a\in v\}$. 
\end{enumerate}
\end{example}

\begin{rem}\hfill
\begin{enumerate}[(1)]
\item In modal logic, given a modal algebra $\alpha:L_\pcal A\to A$,
  one defines a Kripke frame with carrier $SA$ and accessibility
  relation $R_\Box$ given by $vR_\Box u \ \Leftrightarrow \ \forall
  a\in A{.}(\Box a \in v\Rightarrow a \in u)$, see \cite[Def
  5.40]{brv:ml}. To define $R_\Box$ in this way is the same as to give
  $h_A$ as in (4) above, only that $h_A$ is independent of any given
  algebra. More precisely, we obtain $R_\Box$ as $h_A\circ
  S\alpha:SA\to\pcal SA$.
\item Whereas \cite[Def 1]{kkp:calco05} only formulates a condition
  on $h$, (\ref{equ:def-h}) gives us a systematic way of calculating
  $h$ from $\delta$. For finite $A\in\BA$, denoting the units of the
  adjunction by $\iota:\Id\to PS$ and $\eps:\Id\to SP$, we have that
  $h_A$ is given as an arrow in $\Set$ by
\begin{equation}\label{equ:h-finite}
 \Sigma LA
  \stackrel{(\Sigma L\iota_A)^\circ}{\too}\Sigma L\Pi\Sigma A \stackrel{(\Sigma\delta_{\Sigma A})^\circ}{\too}\Sigma\Pi T\Sigma A \stackrel{(\eps_{T\Sigma A})^\circ}{\too} T\Sigma A
\end{equation}
Here we use that $T$ preserves finite sets and hence the arrows above
are isos and we can take their inverse, denoted by $^\circ$. Going
through (\ref{equ:h-finite}) explicitely will yield (4) and (5) above
for finite $A$ and it then turns out that (4) and (5) also work for
all $A$.
\item Note that \eqref{equ:h-finite} is the the inverse of 
\begin{equation}\label{equ:def-delta*2}
\delta^*_A \ = \ T\Sigma A \stackrel{(\eps_{T\Sigma A})}{\too}
\Sigma\Pi T\Sigma A \stackrel{(\Sigma\delta_{\Sigma A})}{\too}
\Sigma L\Pi\Sigma A \stackrel{(\Sigma L\iota_A)}{\too}
 \Sigma LA
\end{equation}
which already appeared as \eqref{equ:def-delta*}.
\end{enumerate}
\end{rem}

\noindent We now come to the main theorem of Part III. Recall Definition~\ref{def:abstract-logic} of a logic for $T$-coalgebras and the global consequence relation \eqref{eq:global-consequence}.

\begin{thm}\label{thm:jons-tars}
  Let $T:\Set\to\Set$ preserve finiteness of sets and weakly preserve
  cofiltered limits. Then $T$ has a sound and strongly complete modal
  logic.
\end{thm} 

\begin{proof}
  Suppose $\Gamma\not\vdash\phi$. Let $A$ be the free $L_T$ algebra
  quotiented by $\Gamma$. By Theorem~\ref{thm:IndPro}, there is a
  $T$-coalgebra on $SA$ such that the injective $\iota_A:A\to PSA$ is
  an $L_T$-algebra morphism.  $\iota_A$ maps all propositions in
  $\Gamma$ to all of $SA$, but $\phi$ only to a proper subset.
  Therefore there is an element in $SA$ satisfying $\Gamma$ and
  refuting $\phi$.
\end{proof}
 
\begin{rem}\hfill
\begin{enumerate}[(1)] 
\item The condition of weak preservation of cofiltered limits is
  elegant, but going back to Theorem~\ref{thm:IndPro} we find that it
  is enough to ask that we can find $h_A:SLA\to TSA$ such that
\begin{equation}\label{eq:delta-h2}
\delta^*_A\circ h =\id_A
\end{equation} 
where $\delta^*$ is as in \eqref{equ:def-delta*2}. It follows from
Theorem~\ref{thm:IndPro} that under \eqref{eq:delta-h2} strong
completeness holds without the conditions of $T$ restricting to finite
sets or weakly preserving filtered colimits. This version of the
theorem was first proved as \cite[Theorem 3]{kkp:calco05}, although
\cite{kkp:calco05} only states the completeness, not the strong
completeness consequence, of the J\'onsson-Tarski-style representation
theorem. Theorem~\ref{thm:jons-tars} extends \cite[Theorem
3]{kkp:calco05} first by the construction of the logic $L_T$ from the
functor $T$ and second by giving a sufficient condition directly in
terms of $T$ for this logic to be strongly complete.
\item The property of logics expressed in the J\'onnson-Tarski-style
  representation theorems \cite[Theorem 3]{kkp:calco05},
  Theorem~\ref{thm:IndPro} and Theorem~\ref{thm:jons-tars}, known as
  canonicity in modal logic, is stronger than strong completeness.  It
  is also worth noting that these representation theorems imply strong
  completeness wrt the global consequence relation which is a stronger
  property in general than strong completeness wrt to the local
  consequence relation. For a comparison of these notions of
  canonicity and strong completeness we refer to Litak
  \cite{litak:phd}.
\item Schr\"oder and Pattinson~\cite{schr-patt:strong-completeness}
  use similar but weaker conditions to prove strong completeness (but
  not canonicity) wrt local consequence. They give a number of
  important examples of such logics for functors $T$ that do not
  restrict to finite sets.
\item The weak preservation of cofiltered limits means, in particular,
  that all projections in the final sequence are onto. The only common
  example of a finite set preserving functor we are aware of that does
  not satisfy this condition is the finite powerset functor, see
  \cite{worrell:final-sequence}. And indeed, standard modal logic is
  strongly complete wrt Kripke frames, but not wrt finitely branching
  ones.
\item The probability distribution functor \cite{vink-rutt:prob-bisim}
  does not preserve finite sets and modal logics for probabilistic
  transition systems, see eg \cite{heif-mong:prob-ml}, are not
  strongly complete. A similar situation occurs for $TX=K\times X$
  where $K$ is an infinite constant.
\item In contrast, we can extend our result to functors $X\mapsto
  (TX)^K$ for infinite $K$ if $T$ preserves finite sets. Indeed, $T^K$
  is a cofiltered limit of the functors $T^{K_i}$ where $K_i$ ranges
  over the finite subsets of $K$. We can now apply the theorem to
  obtain logics $L_{T^{K_i}}$ and then extend the result to the
  colimit of the $L_{T^{K_i}}$ and the limit of the $T^{K_i}$. This
  allows us to include functors such as $(\Pow X)^K\iso \Pow(K\times
  X)$, $K$ infinite (which give rise to labelled transition systems).
\end{enumerate} 
\end{rem}

\section{Conclusion} 
 
\noindent\textbf{Summary } 
The purpose of the paper was to associate a finitary modal logic to a
functor $T$, so that the logic is strongly complete wrt
$T$-coalgebras.  We took up the idea, well-established in domain
theory \cite{abra-jung:dt}, that a logic for the solution of a domain
equation $X\iso TX$ is given by a presentation of the dual $L$ of $T$.
To obtain a logic from $L$, one presents $L$ by operations and
equations and we characterised those functors on a variety that have a
presentation (Theorem~\ref{thm:present-sift}) in Part I. This result
is based on the fundamental role that sifted colimits play in the
category theoretic analysis of universal algebra, see
\cite{arv:varieties}.

\pskip To obtain strong completeness of the logic, we showed in Part
II how to present $L$-algebras as $T$-coalgebras,
Theorem~\ref{thm:IndPro}. This can be considered as the
J{\'o}nsson-Tarski Theorem for $L$-algebras and $T$-coalgebras.

\pskip Part III shows how an arbitrary $T:\Set\to\Set$ gives rise to a
logic $L_T$. By Part I, we know that $L_T$ has a presentation and,
therefore, corresponds to a modal logic given by operations and
equations. Applying the representation theorem of Part II, we obtain
that under additional conditions on $T$, this logic is strongly
complete for $T$-coalgebras.

\pskip An interesting point is that we do not need the assumption that
$T$ is finitary. This assumption is powerful when working with
$T$-algebras, but it is much less so for $T$-coalgebras. Similarly, we
do not need that $T$ preserves weak pullbacks. Each of these
assumptions would exclude fundamental examples.

\pskip\textbf{Further work } 
An important aspect of this work is that it makes use of the notion of
the presentation of a functor in order to separate syntax and
semantics. For example, the strong completeness proof of
Theorem~\ref{thm:jons-tars} is conducted---via
Theorem~\ref{thm:IndPro}---in terms of abstract category theoretic
properties of the logic $(L,\delta)$ and is independent of a choice of
concrete presentation. This approach was also used in
\cite{kurz-rosi:calco07}, which proves a Goldblatt-Thomason style
theorem for coalgebras, and in \cite{kurz-leal:mfps09}, which compares
and translates logics given by predicate liftings and Moss's
coalgebraic logic. This is based on the observation that the notion of
a coalgebraic logic $(L,\delta)$ also accounts for Moss's logic and
makes it amenable to a study via Stone duality. This idea was also
used in \cite{kkv:aiml08} to give a completeness proof for an
axiomatisation of the finitary Boolean version of Moss's logic.

\pskip Another important feature of our approach, which goes back to
\cite{bons-kurz:fossacs06}, is that it is modular in the sense that
the presentation of $\Alg(L)$ is obtained by composing a presentation
of the base category with a presentation of $L$, see
Theorem~\ref{thm:present-fun}.  This can be extended to a formalism
that allows to compose the presentation of $L_1L_2$ from presentations
of $L_1$ and of $L_2$ \cite{kurz-petr:cmcs08-j}. This requires to move
to many-sorted universal algebra and \cite{kurz-petr:cmcs08-j} also
investigates further applications of the many-sorted generalisation to the
semantics of first-order logic and presheaf models of name-binding.

\pskip \cite{kurz-petr:domains9} exploits that the nominal algebra
\cite{gabb-math:nomuae} of Gabbay and Mathijssen and the nominal
equational logic \cite{clou-pitt:nom-equ-log} of Clouston and Pitts
gives rise to theories which correspond to sifted colimit preserving
monads on the category $\mathsf{Nom}$ of nominal sets and can thus be
viewed as equational theories of many-sorted set-based universal
algebra.

\pskip Myers \cite{myers:phd} extends our work on presentations of
Part II to other notions of presentations of functors on varieties
and, importantly, starts the systematic investigation of connecting
properties of presentations with properties of algorithms checking
eg for bisimilairty of process expressions.

\end{document}